%% file: obseq.tex
\documentclass[compsoc, conference, letterpaper, 10pt]{IEEEtran}

\usepackage{proof}
\usepackage{amssymb}
\usepackage{amsmath}
\usepackage{stmaryrd}
\RequirePackage{txfonts}
\DeclareMathAlphabet{\mathsl}{OT1}{cmr}{m}{sl}

\usepackage{color}
\usepackage{url}
\usepackage{ifpdf}
\usepackage{lineno}
\usepackage{balance}
\usepackage{graphicx}
\usepackage{lipsum}
\usepackage{wrapfig}
\usepackage{xspace}
\RequirePackage{txfonts}
\usepackage{multirow}
\usepackage{booktabs}

\ifpdf
\usepackage[colorlinks=true,linkcolor=blue,citecolor=blue]{hyperref}
\fi
\usepackage{cite}
\usepackage[multiple]{footmisc}

\ifpdf
\usepackage[colorlinks=true,linkcolor=blue,citecolor=blue]{hyperref}
\fi
\usepackage{cite}
\usepackage[multiple]{footmisc}

\long\def\comment#1{\relax}
\long\def\omitthis#1{\relax}

\newcommand{\eg}{{\em e.g.}}
\newcommand{\ie}{{\em i.e.}}

\newcommand{\Pscr}{\mathcal{P}}

\newcommand{\Oscr}{\mathcal{O}}
\newcommand{\Cscr}{\mathcal{C}}
\newcommand{\Lscr}{\mathcal{L}}
\newcommand{\Nscr}{\mathcal{N}}
\newcommand{\Sscr}{\mathcal{S}}
\newcommand{\Rscr}{\mathcal{R}}
\newcommand{\Tscr}{\mathcal{T}}

\newcommand{\Gscr}{\mathcal{G}}

\newcommand{\Vscr}{\mathcal{V}}
\newcommand{\TC}{\mathcal{TC}}

\newcommand{\IK}{\mathcal{IK}}
\newcommand{\DC}{\mathcal{DC}}

\newcommand{\EQ}{\mathcal{EQ}}

\newcommand\termEqApprox{\mathsf{termEqApprox}}
\newcommand\canEq{\mathsf{canEq}}
\newcommand\eqCheck{\mathsf{eqCheck}}

\newcommand\sk{\ensuremath{\mathsf{sk}}\xspace}
\newcommand\pk{\ensuremath{\mathsf{pk}}\xspace}
\newcommand\enc{\ensuremath{\mathsf{e}}\xspace}
\newcommand\n{\ensuremath{\mathsf{n}}\xspace}
\renewcommand\k{\ensuremath{\mathsf{k}}\xspace}

\renewcommand\c{\ensuremath{\mathsf{c}}\xspace}
\renewcommand\t{\ensuremath{\mathsf{t}}\xspace}
\newcommand\pl{\ensuremath{\mathsf{pl}}\xspace}

\renewcommand\r{\ensuremath{\mathsf{r}}\xspace}
\newcommand\m{\ensuremath{\mathsf{m}}\xspace}
\renewcommand\v{\ensuremath{\mathsf{v}}\xspace}
\newcommand\p{\ensuremath{\mathsf{p}}\xspace}
\newcommand\ms{\ensuremath{\mathsf{ms}}\xspace}
\newcommand\Syms{\ensuremath{\mathsf{Syms}}\xspace}
\newcommand\dom{\ensuremath{\mathsf{dom}}\xspace}

\newcommand\tv{\ensuremath{\mathsf{tt}}\xspace}

\newcommand\cur{\ensuremath{\mathsf{cur}}\xspace}
\newcommand\floor{\ensuremath{\mathsf{floor}}\xspace}
\newcommand\ceiling{\ensuremath{\mathsf{ceiling}}\xspace}
\newcommand\tr{\ensuremath{\mathsf{tr}}\xspace}
\newcommand\tc{\ensuremath{\mathsf{tc}}\xspace}
\newcommand\new{\ensuremath{\mathsf{new}}\xspace}
\newcommand\myif[3]{\ensuremath{\mathsf{if}~#1~\mathsf{then}~#2~\mathsf{else}~#3}\xspace}

\newcommand\key[1]{\ensuremath{\mathsf{#1}}\xspace}

\newcommand\ds{\ensuremath{\mathsf{ds}}\xspace}

\newcommand\tG{\mathsf{tG}}
\newcommand\ks{\mathsf{keys}}

\renewcommand\sb{\mathsf{sb}}
\newcommand\ssb{\mathsf{ssb}}
\newcommand\mtssb{\emptyset}
\newcommand\css{\mathsf{css}}
\newcommand\sgen{\mathsf{sgen}}
\newcommand\sgenB{\mathsf{sgenB}}
\newcommand\addKeys{\mathsf{addKeys}}
\newcommand\add{\mathsf{add}}

\newcommand\symk{\mathsf{symk}\xspace}

\newcommand\dc{\mathsf{dc}}
\newcommand\Neq{\mathsf{neq}}
\newcommand\Eq{\mathsf{eq}}

\newcommand\sym{\mathsf{sym}}

\newcommand\checkSubst{\mathsf{checkSubst}}
\newcommand\checkBnd{\mathsf{chkBnd}}
\newcommand{\tup}[1]{\langle#1\rangle}

\renewenvironment{paragraph}[1]{\smallskip\noindent{\it #1}~}

\usepackage{tweaklist}
\renewcommand{\enumhook}{
  \setlength{\topsep}{0em}
  \setlength{\itemsep}{0em}
  \setlength{\partopsep}{0em}
  \setlength{\parskip}{0em}
  \setlength{\parsep}{0em}
  \addtolength{\leftmargin}{-0.75em}
}
\renewcommand{\itemhook}{ 
  \setlength{\topsep}{0em}
  \setlength{\itemsep}{0em}
  \setlength{\partopsep}{0em}
  \setlength{\parskip}{0em}
  \setlength{\parsep}{0em}
  \addtolength{\leftmargin}{-0.75em}
}

\newcommand\lra{\longrightarrow}

\newtheorem{theorem}{\bf Theorem}[section]
\newtheorem{definition}[theorem]{\bf Definition}
\newtheorem{example}[theorem]{\bf Example}
\newtheorem{lemma}[theorem]{\bf Lemma}
 \newtheorem{proposition}[theorem]{\bf Proposition}

\newenvironment{proof}{\paragraph{Proof:}}{\hfill$\square$}
\newcommand\bij{\mathsf{bij}}
\newcommand\BB{\mathcal{BB}}

\newcommand\TL{\mathcal{TL}}
\newcommand\tl{\mathsf{pl}}
\begin{document}

\title{Symbolic Timed Observational Equivalence}


\author{\IEEEauthorblockN{Vivek Nigam\IEEEauthorrefmark{1}\IEEEauthorrefmark{3} \quad Carolyn Talcott\IEEEauthorrefmark{2} \quad Abra\~ao Aires Urquiza\IEEEauthorrefmark{1}}
\IEEEauthorblockA{\IEEEauthorrefmark{1}Federal University of Para\'iba, Jo\~ao Pessoa, Brazil,
\texttt{\{vivek.nigam,abraauc\}@gmail.com}
}
\IEEEauthorblockA{\IEEEauthorrefmark{2}SRI International, Melno Park, USA, 
\texttt{clt@csl.sri.com}
}
\IEEEauthorblockA{\IEEEauthorrefmark{3}fortiss, Munich, Germany\\ 
}
}


\maketitle
\begin{abstract}
Intruders can infer properties of a system by measuring the time it takes for the system to respond to some request of a given protocol, that is, by exploiting time side channels. These properties may help intruders distinguish whether a system is a honeypot or concrete system helping him avoid defense mechanisms, or track a user among others violating his privacy. Observational equivalence is the technical machinery used for verifying whether two systems are distinguishable. Automating the check for observational equivalence suffers the problem of state-space explosion problem. Symbolic verification is used to mitigate this problem allow for the verification of relatively large systems. This paper introduces a novel definition of timed observational equivalence based on symbolic time constraints. Protocol verification problems can then be reduced to problems solvable by off-the-shelf SMT solvers. We implemented such machinery in Maude and carry out a number of preliminary experiments demonstrating the feasibility of our approach.
\end{abstract}

\section{Introduction}
\label{intro}
\input{intro}

\section{Examples}\label{examples}
\input{examples}

\section{Term Language}
\label{sec:bas-language}
\input{basic-language}

\section{Observational Equivalence}
\label{sec:obs-eq}
\input{sec-obs}

\section{Experimental Results}
\label{sec:exp}
 \input{experiments}

\section{Related and Future Work}
\label{related}
\label{sec:related}

\input{related}


\newpage

\bibliographystyle{abbrv}
\bibliography{clt}

\appendix

\subsection{Constraint Solver Algorithm}
We define the function $\sgen$ that will compute a finite set of such pairs such that every derivable instance is
covered by some pair.  In particular, $\sgen$ takes as input:
\begin{itemize}
  \item $\ms$ -- the symbolic term to be generated;
  \item $\IK$ -- the intruder knowledge;
  \item $\ssb$ -- a symbol substitution of symbols to symbolic terms, initially empty ($\mtssb$). Intuitively, these represent the symbols which are no longer constrained in the set of derivability constraints, but should be replaced by a particular symbolic term; 
  \item $\DC$ -- a set of derivability constraints which specify the symbols in $\ms$ and $\IK$ and the range of $\ssb$.
\end{itemize}

$\sgen(\ms,\IK,\ssb,\DC)$ returns a set solutions: pairs
\[
  \tup{\ssb_1, \DC_1}, \ldots, \tup{\ssb_n, \DC_n}
\]
where for $1 \leq i \leq n$, $\ssb_i$ is the symbol substitution with the symbols that have been resolved and $\DC_i$ a set of derivability constraints refining $\DC$.

$\sgen(\ms,\IK,\ssb,\DC)$ is defined as follows:\footnote{In our maude code, it is implemented as the function sGen1 in constraints.maude.}
\footnote{The freshly symbols are constrained only if necessary, to allow for the possibility that it only appears in a encryption that can be matched in $\IK$.}
\begin{itemize}
  \item $\ms$ is a symbol: 
  if $\ms \in \DC$ or $\ms$ has guessable type, then return $\{\ssb,\DC\}$,  
  otherwise the symbol must be constrained,
  return $\{\ssb,\DC \dc(\ms, \IK)\}$;
  \item $\ms$ is a nonce: if $\ms \in \IK$, then return $\{\ssb,\DC\}$, otherwise return the empty set (no solution);
  \item $\ms$ is a guessable constant, such as, player name, public keys, etc:  return $\{\ssb,\DC\}$;
  \item $\ms = \{\ms_1, \ldots, \ms_n\}$: iterate through the tuple elements accumulating symbol substitutions and symbol constraints;
  \item $\ms = \enc(\ms_1,\ms_2)$: There are two possibilities:
  \begin{itemize}
    \item The intruder is able to derive $\ms$: This is done by simply calling $\sgen$ on the tuple $\{\ms_1,\ms_2\}$;
    \item The intruder possesses an encryption term $\ms_1$ that is unifiable with $\ms$. For this, we proceed in two steps. First, we treat symbols as variables and attempt to unify $\ms$ with terms in $\IK$, checking for cycles (occur-checks). Once all unifications (symbol substitutions) are found, we only keep the unifications that are consistent with $\DC$. This is done by accumulating the constraints returned by the function $\checkSubst(\ssb_i,\ssb,\ds)$, defined below, for each
unification $\ssb_i$.  Note that the domains of $\ssb_i$ and $\ssb$ are disjoint.
\end{itemize}
\end{itemize}

\noindent
The function $\checkSubst(\ssb_0,\ssb,\DC)$
the set of consistent refinements of $\{\ssb_0,\ssb,\DC\}$.
It processes a binding, $\sym_i \mapsto \ms_i$, in $\ssb_0$ in the context of a pair $\{\ssb,\DC\}$ producing the set of solutions $\css = \{\ssb_j,\DC_j\}, 1 \le j \le k$ that refine $\{\ssb,\DC\}$, and bind $\sym_i$ to an
instance of $\ms_i$ using $\checkBnd(\sym_i,\ms_i,\ssb,\DC)$. 
The next binding is processed
in the context of each element of $\css$ produced
by processing the preceeding bindings.  

The function $\checkBnd(\sym,\ms,\ssb,\DC)$ works
as follows
\begin{itemize}
  \item if $\sym$ is not constrained in $\DC$, then  
  $\checkBnd$ 
returns $\{(\sym \mapsto \ms) \ssb', \DC'\}$ where $\ssb'$ is the result of applying $\sym \mapsto \ms$ to the range of $\ssb$,
and $\DC'$ is the result of applying $\sym \mapsto \ms$
to the terms in $\DC$;

 \item if $\sym$ is constrained in $\DC$, and $\ms$ 
 is not a symbol then, then $\checkBnd$ must ensure
 that $\ms$ is derivable under the accumulated
 constraints. This can be done by 
adding the binding $\sym \mapsto \ms$ to each
solution returned by $\sgen(\ms,\IK,\ssb_1,\DC_1)$.
Here $\DC = \DC_0 dc(\sym,\IK)$, $\ssb_1$ is the result of applying $\sym \mapsto \ms$ to the range of $\ssb$,
and $\DC_1$ is the result of applying $\sym \mapsto \ms$
to the terms in $\DC_0$;
 
\item 
if $\sym$ is constrained in $\DC$, and $\ms$ is a symbol 
that is not constrained in $\DC$ then return
$\{(\sym \mapsto \ms) \ssb_1,\DC_1 dc(\ms,\IK)\}$ where $\ssb_1$, $\DC_1, \IK$
are as above.

\item 
if $\sym$ is constrained in $\DC$, and $\ms$ is a symbol 
that is also constrained in $\DC$ then 
return
$\{(\sym \mapsto \ms) \ssb_1,\DC_1 dc(\ms,\IK')\}$ where $\ssb_1$, $\DC_1$ are as above and $\IK'$ is the constraint
of the earlier of $\sym,\ms$.   One symbol is earlier
than another if if was generated in an earlier step
in the protocol exeution.  Intruder knowledge patterns increase over time, thus we are effectively restricting
to the lesser of the knowledge sets.
\end{itemize}

\subsection{Proof of Theorem~\ref{th:basic-obs}}

\begin{lemma}
\label{lemma:sym-approx}
  Let $\ms$ be a term in $\Oscr$ and $\sym'$ a symbol in $\Oscr'$. 
  $\ms \preceq_{\Oscr,\Oscr'} \sym'$ if and only if $\symDer(\sym',\ms,\DC,\DC')$.
\end{lemma}
\begin{proof}
  We prove by induction on the greatest height, $h$, of the symbols in $\ms$ in the  the dependency graph of the (acyclic) $\DC$. In the following assume $\dc(\sym', \Sscr') \in \DC$.
 \begin{enumerate}
   \item
   \label{proofitem:sym-approx-base}
    \textbf{Base Case:} If $h = 0$, that is, there are no symbols in $\ms$, that is, it is ground. We proceed by induction on $\ms$. 
   \begin{enumerate}
     \item \textbf{Base Case 1:} If $\ms$ is a guessable, then $\ms \preceq_{\Oscr,\Oscr'} \sym'$ and $\symDer(\sym',\ms,\DC,\DC')$ are both true;
     \item \textbf{Base Case 2:} If $\ms$ is a nonce, then there are two cases. Either, this nonce comes from a bijection, which means that $\ms \preceq_{\Oscr,\Oscr'} \sym'$ if and only if $\n \in \DC(\sym')$ if and only if $\n \in \Sscr'$ if and only if $\symDer(\sym',\ms,\DC,\DC')$. Otherwise, $\n \in \DC(\sym')$ if and only if $\symDer(\sym',\ms,\DC,\DC')$;
     \item \textbf{Inductive Case:} If $\ms$ is a tuple $\tup{\m_1, \ldots, \m_n}$, then $\ms \preceq_{\Oscr,\Oscr'} \sym'$ if and only if $\m_i \preceq_{\Oscr,\Oscr'} \sym'$ if and only if (by IH) $\symDer(\sym',\m_i,\DC,\DC')$ for all $1 \leq i \leq n$.
    \item \textbf{Inductive Case:} If $\ms$ is an encryption $\enc(\m,\k)$, then either it is a black-box and follows the same reasoning as with the nonce case. Otherwise we appeal to the inductive hypothesis as with the tuple case.
   \end{enumerate}
   \item $\textbf{Inductive Case:}$ If $h = n + 1$ with $n \geq 0$, then $\ms$ contains some symbols. We proceed by induction on the size of $\ms$. Most of the cases have the same reasoning as before in the proof of the base case (Case~\ref{proofitem:sym-approx-base}) with the exception of the following base case:
   \begin{enumerate}
     \item \textbf{Base Case:} If $\ms$ is a symbol $\sym$. Let $\dc(\sym,\Sscr)$. All ground terms in $\DC(\sym,\Sscr)$ are built using the terms in $\Sscr$.
     As $\DC$ is acyclic, all symbols in $\Sscr$ have height of at most $n$. Thus we appeal to the IH to prove this case.
   \end{enumerate}
 \end{enumerate}
\end{proof}

\begin{proof}
We proceed by induction on the size of $\ms'$. 
\begin{itemize}
  \item \textbf{Base Case:} $\ms'$ is a symbol $\sym'$. The matching substitution is $\sym'\mapsto \ms$, which is handled by Lemma~\ref{lemma:sym-approx}. 

  \item \textbf{Other Base Cases} $\ms' = \n'$ is a nonce. Then $\ms \preceq_{\Oscr,\Oscr'} \n'$ if and only if $\ms = \bij[\n']$ if and only if $\termApprox(\ms,\ms',\Oscr,\Oscr')$. If it is a guessable, then $\ms$ has to be the same guessable and $\termApprox(\ms,\ms',\Oscr,\Oscr')$. Similar when $\ms'$ is a key.

  \item \textbf{Inductive Case:} $\ms' = \enc(\ms_1',\k')$ is an encrypted term. Then it is either black-boxed, in which case the proof is similar to the case when it is a nonce. Otherwise, we appeal to the IH on the smaller terms $\ms_1'$ and $\k'$;

  \item \textbf{Inductive Case:} $\ms = \{\ms_1',\ldots,\ms_n'\}$, then we appeal to the IH on the smaller terms $\ms_i'$.
\end{itemize}
\end{proof}

\subsection{Proof of Theorem~\ref{th:branching-obs}}

\begin{lemma}
\label{lemma:equality}
  Let $\EQ = \EQ_1 \cup \EQ_2$ be a set of comparison constraints, where $\EQ_1$ contains only equality constraints and $\EQ_2$ inequality constraints. Let $\theta$ be the most general unifier of all constraints in $\EQ_1$. For all ground terms $\m$, $\m \in \DC(\ms)|_\EQ$ if and only if $\m \in \DC(\theta[\ms])|_{\EQ_2}$ 
\end{lemma}
\begin{proof}
  $\m \in \DC(\ms)|_\EQ$ if and only if the pattern match $\ssb$ of $\m$ and $\ms$  does not falsify any constraint in $\EQ$ if and only if $\ssb$ is an instance of $\theta$ (as it is the m.g.u.) and $\ssb$ does not falsify any constraint in $\EQ_2$ if and only if $\m \in \DC(\theta(\ms))|_{\EQ_2}$.
\end{proof}

\begin{lemma}
\label{lemma:case2}
There is ground term $\m \in \DC(\ms_1)|_\EQ$ and $\m \in \DC(\ms_2)|_\EQ$ with the same witnessing substitution $\theta$ if and only if $\canEq(\ms_1,\ms_2,\DC,\EQ)$.
\end{lemma}
\begin{proof}
  $\m \in \DC(\ms_1)|_\EQ$ and $\m \in \DC(\ms_2)|_\EQ$ with the same substitution $\theta$ if and only if $\theta(\ms_1) = \theta(\ms_2) = \m$ if and only if $\ms_1$ and $\ms_2$ can be unified by $\theta$ and $\theta$ satisfies all constraints in $\EQ$ and for all $\sym_i \mapsto \m_i^* \in \theta$, $\m_i \in \DC(\sym_i)$ (Definition of $\DC(\sym_i)$ membership) if and only if $\canEq(\ms_1,\ms_2,\DC,\EQ)$.
\end{proof}

\begin{lemma}
\label{lemma:case1}
Let $\EQ$ be a finite set of inequality constraints only. Assume $\DC \vDash \EQ$.
  For all ground terms $\m$ we have that $\m \in \DC(\ms)|_{\EQ} \Rightarrow \m \in \DC'(\ms')$ if and only if $\m \in \DC(\ms) \Rightarrow \m \in \DC'(\ms)$
\end{lemma}
\begin{proof}
The reverse direction is immediate. We prove the forward direction. 
Assume (1) $\m \in \DC(\ms)|_{\EQ} \Rightarrow \m \in \DC'(\ms')$ for all ground terms $\m$ and assume that (2) $\m_1 \in \DC(\ms)$. We show that $\m_1 \in \DC'(\ms)$.

We proceed by induction on the size of $\ms$. The interesting case is when $\ms = \sym$ is a symbol which means that $\ms' = \sym'$ has to be a symbol. Otherwise, it is easy to construct a term $\m_2 \in \DC(\ms)|_{\EQ}$ such that $\m_2 \notin \DC'(\ms')$. (For example, a very large tuple of guessables.) 

Let the matching symbol substitution $\ssb = [\sym \mapsto \m_1]$. There are two cases:
\begin{itemize}
  \item $\ssb$ does not some $\EQ$ false, then $\m_1 \in \DC(\ms)|_\EQ$ and thus $\m_1 \in \DC(\ms)$ by (1).

  \item $\ssb$ falsifies some constraint in $\EQ$. Assume by contradiction that $\m_1 \notin \DC'(\ms)$.  Since $\m_1 \in \DC(\ms)$, an arbitrary large tuple $\tup{\m_1, \ldots, \m_1} \in \DC(\ms)$. However, since $\m_1 \notin \DC'(\ms)$, then (3) $\tup{\m_1, \ldots, \m_1} \notin \DC'(\ms)$. Pick a large tuple such that no constraint in $\EQ$ is falsified. (Recall that all constraints in $\EQ$ are inequality constraints.) Then $\tup{\m_1, \ldots, \m_1} \in \DC(\ms)|_\EQ$. From (1), we get $\tup{\m_1, \ldots, \m_1} \in \DC'(\ms)$ yielding a contradiction with (3). Thus $\m_1 \in \DC'(\ms)$.
\end{itemize}
\end{proof}

\begin{proof}
  $\ms \preceq_{\Oscr,\Oscr'} \ms'$ if and only if 
\begin{itemize}
  \item $\DC \nvDash \EQ$ if and only if $\termEqApprox(\ms,\ms',\Oscr,\Oscr')$ by Lemma~\ref{lem:eqCheck};
  \item or $\DC \vDash \EQ$ and $\DC' \vDash \EQ'$ and for all $\m \in \DC(\ms)|_\EQ$, we have that $\m \in \DC'(\ms')|_{\EQ'}$. By Lemma~\ref{lemma:equality}, we have that $\m \in \DC(\ms)|_\EQ$ if and only if  $\m \in \DC(\theta[\ms])|_{\EQ_2}$, where $\theta$ is the mgu of the equality constraints in $\EQ$ and $\EQ_2$ are the inequality constraints. Similarly $\m \in \DC'(\ms')|_{\EQ'}$ if and only if $\m \in \DC'(\theta'[\ms'])|_{\EQ_2'}$. Thus we only need to consider the inequality constraints $\EQ_2$ and $\EQ_2'$.

   By contraposition, we attempt to find a term $\m \in \DC(\ms)|_{\EQ_2}$ such that $\m \notin \DC'(\ms')|_{\EQ_2'}$. This leads to two possibilities, where $\ssb_1,\ssb_1'$ be the symbol substitution such that $\ssb_1[\ms] = \ssb_1'[\ms'] = \m$. This can only exists if $\ssb_1[\theta[\ms']] = \m$ where $\theta$ is the matching substitution of the symbols in $\ms'$ to terms in $\ms$ (as in Definition~\ref{def:term-approx}). Otherwise, the terms $\ms$ and $\ms'$ cannot derive the same terms.
  \begin{itemize}
   \item $\ssb_1'$ does not falsify a constraint in $\EQ_2'$ if and only if $\m \notin \DC'(\ms')$ if and only if by Lemma~\ref{lemma:case1} and Theorem~\ref{th:basic-obs} 
   $\termApprox(\ms,\ms',\Oscr,\Oscr')$ is false;
   \item $\ssb_1'$ falsifies a constraint in $\Neq(\ms_1,\ms_2) \in \EQ_2'$ if and only if $\ssb_1[\theta(\ms_1)] = \ssb_1[\theta(\ms_2)]$ if and only if $\canEq(\theta(\ms_1),\theta(\ms_2),\DC,\EQ)$ is true (Lemma~\ref{lemma:case2}). 
  \end{itemize}
\end{itemize}
\end{proof}

\end{document}

%% file: intro.tex
Time side channels can be exploited by intruders in order to infer properties of systems, helping them avoid defense mechanisms, and track users, violating their privacy. For example, honeypots are normally used for attracting intruders in order to defend real systems from their attacks. However, as honeypots run over virtual machines whereas normal client systems usually do not, it takes longer for a honeypot to respond to some protocol requests. This information can be used by the attacker to determine which servers are real and which are honeypots.
For another example, passports using RFID mechanisms have been shown to be vulnerable to privacy attacks. An intruder can track a particular's passport by replaying messages of previous sessions and measuring response times.

The formal verification of such properties is different from usual reachability based properties, such as secrecy, authentication and other correspondence properties. In the verification of reachability properties, one searches for a  trace that exhibits the flaw, \eg, the intruder learning a secret. In attacks 
such as the ones described above, one searches instead for behaviors that can distinguish two system, \eg, a behavior that can be observed when interacting with one system, but that cannot be observed when interacting with the other system. That is, to check whether the systems are \emph{observationally distinguishable}. This requires reasoning over sets of traces.

Various notions of \emph{observational equivalence} have been proposed in the programming languages community as well as in concurrent systems~\cite{agha-mason-smith-talcott-96jfp,hofmann16ppdp,milner,gunter} using, for example, logical relations and bisimulation. Observational equivalence has also been proposed for protocol verification notably the work of Cortier and Delaune~\cite{cortier09csf}. A number of properties, \eg, unlinkability and anonymity~\cite{arapinis10csf}, have been reduced to the problem of observational equivalence. As protocol verification involves infinite domains, the use of symbolic methods has been essential for the success of such approaches.


The contribution of this paper is three-fold:
\begin{itemize}
\item \textbf{Symbolic Timed Observational Equivalence:} We propose a novel definition of timed equivalence over timed protocol instances~\cite{nigam16esorics}.
Timing information, \eg, duration of computation, is left symbolically and can be specified in the form of time constraints relating multiple time symbols, \eg, $\tv_1 \geq \tv_2 + 10$;

\item \textbf{SMT Solvers for proving Time Observational Equivalence:} SMT solvers are used in two different ways. We specify the operational semantics of timed protocols using Rewriting Modulo SMT~\cite{rocha-2012}. 
Instead of instantiating time symbols with concrete values, in Rewriting Modulo SMT, a configuration of the system is symbolic and therefore may represent an unbounded number of concrete configurations. Rewriting of a symbolic configuration is only allowed if the set of (time) constraints in the resulting state is satisfiable. SMT-Solvers are used to perform this check. This means not only that there are a finite number of symbolic traces starting from a given configuration, but also reduces considerably the search space needed to enumerate these traces. We demonstrate this with experiments. 

The second application of SMT-Solvers is on the proof of timed observational equivalence, namely, to check whether the timing of observations can be matched. This check involves the checking for the satisfiability of $\exists \forall$ formulas~\cite{duterte15smt}.

\item \textbf{Implementation}: Relying on the Maude~\cite{clavel-etal-07maudebook}  support for Rewriting Modulo SMT using the SMT-solvers CVC4~\cite{barrett11cvc4} or
Yices~\cite{yices}, we implemented in Maude the machinery necessary for enumerating symbolic traces. However, as checking for the satisfiability of $\exists \forall$ formulas~\cite{duterte15smt} is not supported by Maude, we integrate our Maude machinery with the SMT solver Yices~\cite{duterte15smt}. We carry out some proof-of-concept experiments demonstrating the feasibility of our approach.
\end{itemize}


Section~\ref{examples} describes some motivating examples on how intruders can using time side channels for his benefit. 
We introduce the basic symbolic language in Section~\ref{sec:bas-language} and the timed protocol language in Section~\ref{subsec:basic-symb-op-sem}. Section~\ref{sec:obs-eq} introduces symbolic timed observational equivalence describing how to prove this property. Section~\ref{sec:exp} describes our implementation architecture and the experiments carried out. Finally, in Section~\ref{sec:related}, we conclude by commenting on related and future work.

Some missing proofs are shown in the Appendix.

%% file: examples.tex
We discuss some motivating examples illustrating how intruders can exploit time side channels of protocols. 

\paragraph{Red Pill}
Our first example is taken from~\cite{HoBBP14woot}. The attack is based on the concept of \emph{red pills}. The overall goal of the attacker is to determine whether some system is running on a virtual machine or not. As honeypots trying to lure attackers normally run on virtual machines, determining if a system is running on a virtual machines or not gives an attacker one means to avoid honeypots~\cite{HoBBP14woot}. The system running in a virtual machine or a concrete machine  follow exactly the same protocol. 

When an application connects to the malicious server,
the server first sends a \emph{baseline} request followed by 
a \emph{differential} request.  The time to respond to the
baseline request is same whether running in a virtual machine
or not and is used for calibration.  The time to respond to the
differential request is longer when executed in a virtual machine.
When not taking time into account, the set of traces for this
exchange is the same whether the application is running on a virtual machine or not. However, if we also consider the time to respond to the
two requests, the timed traces of applications running on virtual
machines can be distingushed from those of applications running
on native hardware.

\paragraph{Passport RFID}
Our second example comes from work of Chothia and Smirnov 
\cite{chothia10fc} investigating the security of \textit{e-passports}.
These passports contain an RFID tag that, when powered, broadcast
information intended for passport readers.  Also, once powered,
e-passport broadcasts can't be turned off.
Chothia and Smirnov identified a flaw in one of the passport’s protocols that makes it possible to trace the movements of a particular passport, without having to break the passport’s cryptographic key. In particular, if the attacker records one session between the passport and a legitimate reader, one of the recorded
messages can be replayed to  distinguish that passport from other
passports.  Assuming that the target carried their passport on them, an attacker could place a device in a doorway that would detect when the target entered or left a building. 
In the protocol, the passport receives an encryption and a mac 
verifying the integrity of the encryption. The protocol first
checks the mac, and reports an error if the check fails.
If the mac check succeeds, it checks the encryption.
This will fail if the encryption isn't fresh.  
When the recorded encryption, mac pair is replayed to the
recorde passport, the mac check will succeed but the encryption check will fail, while the mac check will fail when carried out by any
other passport as it requires a key unique to the passport.
The time to failure is significantly longer for the targeted passport
than for others, since only the mac check is needed and it is faster.

\paragraph{Anonymous Protocol} Abadi and Fournet~\cite{abadi04tcs} proposed an anonymous group protocol where members of a group can communicate within each other without revealing that they belong to the same group. A member of a group broadcasts a message, $m$, encrypted with the shared group key. Whenever a member of a group receives this message, it is able to decrypt the message and then check whether the sender indeed belongs to the group and if the message is directed to him. In this case, the receiver broadcasts an encrypted response $m'$. 

Whenever a player that is not member of the group receives the message $m$, it does not simply drop the message, but sends a decoy message with the same shape as if he belongs to the group, \ie, in the same shape as $m'$. In this way, other participants and outsiders cannot determine whether a two players belong to the same group or not.

However, as argued in \cite{corin04fmse}, by measuring the time when a response is issued, an intruder  can determine whether two players belong to the same group. This is because decrypting and generating a response take longer than just sending a decoy message.

%% file: basic-language.tex



The basic term language contains usual cryptographic operators such as encryption, nonces, tuples.
More precisely the term language is defined by the following grammar. We assume given text constants, $\Tscr$ and player names $\Pscr$. We also assume a countable set of nonces, $\Nscr$, and of symbols, $\Syms$, as well as a countable number of sorted variables, $\Vscr$, where $\Nscr, \Syms$ and $\Vscr$ are disjoint. Below $\v_p$ represents a variable of sort player. 
\[
  \begin{array}{l@{~}l@{\quad}l@{\quad}l}
  \multicolumn{2}{l}{\textbf{Basic Constants:}}\\
      \c := &  \t \in \Tscr & \textrm{Text Constants}\\
      & \mid \p \in \Pscr & \textrm{Player Names} \\
      & \mid \n \in \Nscr & \textrm{Nonces}\\[2pt]
     \multicolumn{2}{l}{\textbf{Keys:}}\\
     \k :=  & \mid \symk & \textrm{Symmetric key}\\
      & \mid \pk(\p) \mid \pk(\v_p) & \textrm{Public key of a player}\\
      & \mid \sk(\p) \mid \sk(\v_p) & \textrm{Secret key of a player}\\[2pt]
     \multicolumn{2}{l}{\textbf{Symbols:}}\\
     \sym :=  & \mid \sym \in \Syms & \textrm{Symbol}\\[2pt]
     \multicolumn{2}{l}{\textbf{Terms:}}\\ 
     \m := & \c & \textrm{Basic constants}\\
      & \mid \k & \textrm{Keys}\\
      & \mid \v \in \Vscr & \textrm{Variables}\\
      & \mid \sym \in \Syms& \textrm{Symbols}\\
      & \mid \enc(\m,\k) & \textrm{Encryption of term \m with key \k }\\
      & \mid \tup{\m_1, \ldots, \m_n} & \textrm{Tuples} 
  \end{array}
\]
A term is \emph{ground} if it does not contain any occurrence of variables and symbols. A term is \emph{symbolic} if it does not contain any occurrence of variables, but it may contain occurrences of symbols. $\ms,\ms_1,\ms_2, \ldots$ will range over symbolic terms. We define $\Syms(\ms)$ as the set of symbols appearing in a symbolic term.

It is possible to add other cryptographic constructions, such as hash, signatures, but in order to keep things simple and more understandable, we only include encryption. As hashes and signatures can be specified using encryption, this is not limiting. Finally, it is easy to extend the results here with fresh keys. These are treated in the same way as nonces, but to keep it simple, we do not include them.


We will use two types of (capture avoiding) substitutions. \emph{Variable substitutions} written $\sb, \sb_1, \sb_2, \ldots$ which are maps from variables to symbolic terms $\sb = [\v_1 \mapsto \ms_1,\v_2 \mapsto \ms_2, \ldots, \v_n \mapsto \ms_n ]$. \emph{Symbol substitutions} written $\ssb, \ssb_1, \ssb_2, \ldots$ mapping symbols to symbolic terms $\ssb = [\sym_1 \mapsto \ms_1,\sym_2 \mapsto \ms_2, \ldots, \sym_n \mapsto \ms_n ]$.

\subsection{Symbolic Term Constraints}
Intuitively, variables are entities that can be replaced by symbolic terms, while a symbol  denotes a (possibly infinite) set of terms.  For example, if the symbol $\sym_j$ can be instantiated by any one of the (symbolic) terms $\{\ms_1, \ldots, \ms_n\}$, then the symbolic term $\enc(\sym,\k)$ represents the set of terms:
\[
\{\enc(\ms,\k) \mid \ms \in \{\ms_1, \ldots, \ms_n\}\}
\]  
Such simple idea has enabled the verification of security protocols, which have infinite search space on ground terms, but finite state space using symbolic terms. 

We formalize this idea by using derivability constraints. Derivability constraints are constructed over minimal sets defined below. 

\begin{definition}
\label{def:minimal}
A set of symbolic messages $\Sscr$ is minimal if it satisfies the following conditions:
\begin{itemize}
  \item $\Sscr$ contains all guessable constants, such as player names and public keys;
  \item $\Sscr$ does not contain tuples;
  \item if $\enc(\ms, \k) \in \Sscr$ if and only if $\k^{-1} \notin \Sscr$ where $\k^{-1}$ is the inverse key of $\k$;
\end{itemize}
Formally, the symbolic terms derivable from $\Sscr$ is the smallest set $\Rscr$ defined inductively as follows:
\begin{itemize}
  \item if $\ms \in \Sscr$ then $\ms \in \Rscr$;
  \item if $\k \in \Rscr$ and $\ms \in \Rscr$, then $\enc(\ms,\k) \in \Rscr$;
  \item if $\ms_1', \ldots, \ms_m' \in \Rscr$, then $\tup{\ms_1',\ldots,\ms_m'} \in \Rscr$;
\end{itemize}
\end{definition}

From two minimal sets, $\Sscr_1$ and $\Sscr_2$, we can construct the minimal set, $\Sscr$, representing the union of $\Sscr_1$ and $\Sscr_2$ by applying the following operations until a fixed point is reached starting from $\Sscr_0 = \Sscr_1 \cup \Sscr_2$:
\begin{itemize}
  \item $\enc(\m,\k)\in \Sscr_i$ and  $\k^{-1} \in \Sscr_i$, then $\Sscr_{i+1} = \Sscr_i \cup \{\m\}$;
  \item $\enc(\m,\k),\k, \k^{-1}\in \Sscr_i$, then $\Sscr_{i+1} = \Sscr_i \setminus\{\enc(\m,\k)\} \cup \{\m\}$;
  \item $\tup{\m_1, \ldots, \m_n} \in \Sscr_i$, then $\Sscr_{i+1} = \Sscr_i \setminus \{\tup{\m_1, \ldots, \m_n}\} \cup \{\m_1, \ldots, \m_n\}$.
\end{itemize}

For example, given the minimal sets:
\[
  \Sscr_1 = \{\symk,\pk\} \textrm{ and } 
  \Sscr_2 = \{\enc(\tup{\enc(\t,\sk), \t},\symk)\}
\]
The minimal set obtained by the union of $\Sscr_1$ and $\Sscr_2$ is:
\[
  \Sscr = \{\symk,\pk,\t,\enc(\t,\sk)\}
\]

We consider two types of constraints on terms: Derivability constraints (Definition~\ref{def:dc}) and comparison constraints (Definition~\ref{def:comparison}). 

\begin{definition}
\label{def:dc}
A derivability constraint has the form $\dc(\sym, \Sscr)$, where $\Sscr$ is minimal. This constraint denotes that $\sym$ can be any (symbolic) term derived from $\Sscr$.
\end{definition}

For example, the  derivability constraint 
\[
 \dc(\sym, \{alice, bob, eve, \pk(alice), \pk(bob),\sk(eve)\}) 
\]
specifies that $\sym$ may be instantiated by, \eg, the terms $\tup{alice, bob}, \tup{alice, \sk(eve)},$ $\enc({alice, \pk(bob)})$, $\enc(\tup{alice, bob}, \pk(bob))$ and so on.

To improve readability (and also reflect our implementation), we will elide in any constraint $\dc(\sym,\Sscr)$ the guessable terms. For example, we write the derivability constraint above simply as $\dc(\sym(1), \{\sk(eve),\})$ as $alice, bob, eve, \pk(alice), \pk(bob)$ are all guessables, namely player names and public keys.

Notice that any $\dc(\sym,\Sscr)$ denotes a infinite number of symbolic terms due to the tupling closure. We will abuse notation and use $\ms \in \dc(\sym, \{\ms_1,\ldots, \ms_n\})$ to denote that the symbolic term $\ms$ is in the set of terms that $\sym$ can be instantiated with. Moreover, we assume that for any given set of derivability constraints $\DC$, there is at most one derivability constraint for any given $\sym$, that is, if $\dc(\sym,\Sscr_1), \dc(\sym,\Sscr_2) \in \DC$, then $\Sscr_1 = \Sscr_2$. We write $\dom(\DC) = \{\sym \mid \dc(\sym,\Sscr) \in \DC\}$. 
We write $\DC(\sym)$ for the derivability constraint for $\sym$ in $\DC$ if it exists. Moreover, we write $\ms \in \DC(\sym)$ if the term can be derived from $\DC(\sym)$.

\begin{definition}
The symbol dependency graph of a given set of derivability constraints $\DC$, written $\Gscr_\DC$, is a directed graph defined as follows:
\begin{itemize}
  \item Its nodes are symbols in $\DC$, that is, $\dom(\DC)$;
  \item It contains the edge $\sym_1 \longrightarrow \sym_2$ if and only if $\dc(\sym_1,\Sscr_1), \dc(\sym_2,\Sscr_2) \in \DC$ and $\Sscr_2$ contains at least one occurrence of $\sym _1$. 
\end{itemize}
\end{definition}

While in general the symbol dependency graph of $\DC$ can be cyclic, our operational semantics will ensure that these graphs are acyclic. 

Consider the following set of derivability constraints:
\[
\DC_0 = \left\{\begin{array}{c}
  \dc(\sym_1, \{\sk(eve)\}),
  \dc(\sym_2, \{\enc(\sym_1, \symk)\}),\\
  \dc(\sym_3, \{\enc(\tup{\sym_2,\sym_1}, \pk(Alice)\}),\\
   \dc(\sym_4, \{\sym_3, \sym_2\})
  \}),
\end{array}\right\} 
\] 
Its dependency graph is the directed acyclic graph (DAG).
\begin{center}
\includegraphics[width=0.25\textwidth]{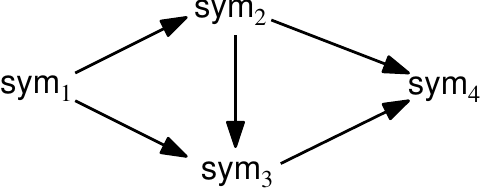}  
\end{center}
Whenever the dependency graph of a set of constraints is a DAG, we classify the set as acyclic. We can compute a topological sort of the DAG in linear time. For example, a topological sort of $\Gscr_{\DC_0}$ is $[\sym_1, \sym_2, \sym_3, \sym_4]$.

Given a set of derivability constraints, we can now formally specify the set of terms that a symbolic term denotes. 

\begin{definition}
\label{def:basic-semantics-terms}
Let $\ms$ be a symbolic term. 
Let $\DC$ be an acyclic set of derivability constraints. 
Assume $\dc(\sym, \Sscr) \in \DC$.
We define the operator $Sub_{\DC}(\sym,\ms)$ as the set of symbolic terms obtained by replacing all occurrences of $\sym$ in $\ms$ by a term $\ms_1 \in \DC(\sym)$. Formally, the set:
\[
\left\{\begin{array}{ll}
  \sigma(\ms) \mid &  \sigma = \{\sym \mapsto \ms_1\}\\
  & \textrm{ is a substitution where $\ms_1 \in \DC(\sym)$}
\end{array}\right\}
\]
Moreover, $Sub_{\DC}(\sym,\Sscr)$ for a set of symbolic terms $\Sscr$ is the set $\bigcup_{\ms \in \Sscr} Sub_\DC(\sym,\ms)$.

Let $\Tscr = [\sym_1,\ldots,\sym_n]$ be any topological sort of the DAG $\Gscr_\DC$. Then the meaning of a symbolic term $\ms$ with respect to $\DC$, written $\DC(\ms)$, is the set obtained by applying $Sub_\DC$ consecutively as follows:

\vspace{-2mm}
\noindent
\begin{small}
\[
  Sub_\DC(\sym_1, Sub_\DC(\sym_2,Sub_\DC(\ldots Sub_\DC(\sym_n,\ms)\cdots))).
\]  
\end{small}
\end{definition}

For example,  $Sub_{\DC_0}(\sym(4),\enc(\sym(4),\pk(bob)))$ is the set of terms:
\[
   \{\enc(\ms_1,\pk(bob)) \mid \ms_1 \in \DC_0(\sym(4))\}
\] 
It contains the terms $\enc(\sym_2,\pk(bob)), \enc(\sym_3,\pk(bob))$, $\enc(\tup{\sym_2,\sym_2} ,\pk(bob))$, $\enc(\tup{\sym_2,\sym_3},\pk(bob)), \ldots$. The set $\DC_0(\enc(\sym_4,\pk(bob)))$ contains the terms 
$\enc(\enc(\sk(eve), \symk),\pk(bob))$, by applying to the term $\enc(\sym_4,\pk(bob))$ the substitution $[\sym_4 \mapsto \sym_2]$ followed by $[\sym_2 \mapsto \enc(\sym_1, \symk)]$  $[\sym_1 \mapsto \sk(eve)]$. 


Notice that for any acyclic set of derivability constraints such that its lowest height symbols (w.r.t. $\Gscr_\DC$) have constraints of the form $\dc(\sym, \Sscr)$ where $\Sscr$ are ground terms, then $\DC(\ms)$ is an (infinite) set of ground terms. This is because the successive application of $Sub_\DC$ will eventually eliminate all symbols.

Given terms $\ms, \ms'$, we describe how to check whether $\ms \in \DC(\ms')$. We first build the matching subsitution $\ssb = \{\sym_1' \mapsto \ms_1, \ldots, \sym_n' \mapsto \ms_n\}$ from symbols in $\ms'$ to (sub)terms in $\ms$. If no such matching subsitution exists, then $\ms \notin \DC(\ms')$. For each $\sym_i'\mapsto \ms_i$, let $\dc(\sym_i, \Sscr_i) \in \DC$. We check whether $\ms_i \in \DC(\sym_i')$ recursively as follows:
\begin{itemize}
  \item If $\ms_i \in \Sscr_i$, return true;
  \item If $\ms_1 = \enc(\ms_2, \ms_3)$, then we check whether $\tup{\ms_2, \ms_3} \in \DC(\sym_i')$;
  \item if $\ms_i = \{\ms_i^1,\ldots, \ms_i^m\}$, then for each $1 \leq j \leq m$, we check whether $\ms_i^j \in \DC(\sym_i')$.
\end{itemize}

\begin{definition}
$\ssb \vDash \DC$ if for each $\sym \mapsto \ms \in \ssb$, $\ms \in \DC(\sym)$.
\end{definition}

The following definitions specify the second type of term constraints called comparison constraints.

\begin{definition}
\label{def:comparison}
  A comparison constraint is either an equality constraint of the form $\Eq(\ms_1,\ms_2)$ or an inequality constraint of the form $\Neq(\ms_1,\ms_2)$. 
\end{definition}

A set $\EQ$ of comparison constraints should be interpreted as a conjunction of constraints. The following definition specifies when it is satisfiable.

\begin{definition}
  Let $\DC$ be a set of derivability constraints and $\EQ$ be a set of comparison constraints. The set $\EQ$ is satisfiable w.r.t. $\DC$, written $\DC \vDash \EQ$, if there is a subsitution $\sigma = \{\sym_1 \mapsto \m_1, \ldots, \sym_n \mapsto \m_n\}$ mapping all symbols $\sym_i$ in $\EQ$ to ground terms in $\DC(\sym_i)$, such that:
  \begin{itemize}
    \item for all equality constraints $\Eq(\ms,\ms') \in \EQ$, $\sigma[\ms] = \sigma[\ms']$;
    \item for all inequality constraints $\Neq(\ms,\ms') \in \EQ$, $\sigma[\ms] \neq \sigma[\ms']$.
  \end{itemize}
\end{definition}

We define the procedure below, $\eqCheck$, for checking whether a set of comparison constraints $\EQ$ is satisfiable. 

\begin{definition}
\label{def:eqCheck}
  Let $\EQ$ be a (finite) set of comparison constraints and $\DC$ a set of derivability constraints. Let $\Eq(\ms_1,\ms_1'), \Eq(\ms_2,\ms_2') \ldots \Eq(\ms_n,\ms_n')$ be all the equality constraints in $\EQ$. Then $\eqCheck(\EQ,\DC)$ is true if and only if
  \begin{enumerate}
    \item There is a unifer $\ssb$ of the terms $\alpha = \tup{\ms_1, \ldots, \ms_n}$ and $\beta = \tup{\ms_1', \ldots, \ms_n'}$ mapping symbols to symbolic terms, that is, $\ssb(\alpha) = \ssb(\beta)$;

    \item For all inequality constraint $\Neq(\ms,\ms') \in \EQ$, $\ssb[\ms] \neq \ssb[\ms']$;

    \item $\ssb$ is consistent with $\DC$ (as done in Section~\ref{subsec:basic-solving}).
  \end{enumerate}
\end{definition}

\begin{lemma}
\label{lem:eqCheck}
  $\DC \vDash \EQ$ if and only if $\eqCheck(\EQ,\DC)$.
\end{lemma}

Moreover, the meaning of a symbolic term should take comparison constraints $\EQ$ into account. That is, it should not be possible to replace a symbol by a term that falsifies some comparison constraint. We extend Definition~\ref{def:basic-semantics-terms} accordingly.

\begin{definition}
Let $\DC$ be an acyclic set of derivability constraints and $\EQ$ a set of comparison constraints. The meaning of a symbolic term $\ms$ w.r.t. $\DC$ and $\EQ$, written $\DC(\ms)|_\EQ$, is the set of terms $\ms' \in \DC(\ms)$ such that there exists a matching substitution $\theta$:
\begin{itemize}
  \item $\theta(\ms) = \ms'$;
  \item For all equality constraints $\Eq(\ms_1,\ms_2) \in \EQ$, $\theta(\ms_1) = \theta(\ms_2)$;
  \item For all inequality constraints $\Neq(\ms_1,\ms_2) \in \EQ$, $\theta(\ms_1) \neq \theta(\ms_2)$.
\end{itemize}  
\end{definition}

For example, $\dc(\sym_1, \{\t_1\}),\dc(\sym_2, \{\t_2\}) \in \DC$ and a set of a single comparison constraint $\EQ = \{\Eq(\sym_1, \sym_2)\}$. The term $\tup{\t_1, \t_2} \in \DC(\tup{\sym_1,\sym_2})$, but $\tup{\t_1, \t_2} \notin \DC(\tup{\sym_1,\sym_2})|_\EQ$. This is because the matching substitution $\theta = \{\sym_1 \mapsto \t_1, \sym_2 \mapsto \t_2\}$ turns the constraint $\Eq(\sym_1, \sym_2)$ false: $\t_1 \neq \t_2$.

\subsection{Symbolic Time Constraints}
Assume a time signature $\Xi$ which is disjoint to the message alphabet $\Sigma$. It contains numbers (real and natural), variables and pre-defined functions. 
\omitthis{In principle, we can use any function supported by the underlying SMT-solver. [clt: redundant]}
\[
\begin{array}{l@{~}l}
  \r_1,\r_2,\ldots & \textrm{A set of numbers;}\\
  \tv_1,\tv_2,\ldots,& \textrm{A set of time variables};\\
  & \textrm{including the special variable $\cur$}\\
  +,-,\times,/, \floor,\ceiling, \ldots & \textrm{A set of pre-defined functions.}
\end{array}  
\]
\emph{Time Expressions} are constructed inductively by applying arithmetic symbols to time expressions. For example $\ceiling((2 + \tv + \cur )/ 10)$ is a Time Expression. The symbols $\tr_1,\tr_2, \ldots$ range over Time Expressions.
We do not constrain the set of numbers and function symbols in $\Xi$. However, in practice, we allow only the symbols supported by the SMT solver used. All examples in this paper will contain SMT supported symbols (or equivalent). Finally, the time variable $\cur$ will be a keyword in our protocol specification language denoting the current global time.

\begin{definition}[Symbolic Time Constraints] Let $\Xi$ be a time signature.
The set of symbolic time constraints is constructed inductively using time expressions as follows: 
Let  
 $\tr_1,\tr_2$ be time expressions, then 
\[
\begin{array}{l}
  \tr_1 = \tr_2, \quad \tr_1 \geq \tr_2 \quad \tr_1 > \tr_2, \quad\tr_1 < \tr_2, \textrm{ and } \tr_1 \leq \tr_2
\end{array}
\]
are Symbolic Time Constraints.
\end{definition}
For example, $\cur + 10 < \floor(\tv - 5)$ is a Time Constraint. 
Time Constraints will range over $\tc, \tc_1, \tc_2, \ldots$.

Intutively, given a set of time constraints $\TC$, each of its models with concrete instantiations for the time variables corresponds to a particular scenario. This means that one single set of time constraints denotes a possibly infinite number of concrete scenarios. For example, the set of constraints $\{\tv_1 \leq 2, \tv_2 \geq 1 + \tv_1\}$ has an infinite number of models, \eg, $[\tv_1 \mapsto 2.1, \tv_2 \mapsto 3.1415]$.

Finally, SMT-solvers, such as CVC4~\cite{barrett11cvc4} and Yices~\cite{duterte15smt}, can check for the satisfiability of a set of time constraints.  

\input symbolicCC

\section{Timed Protocol Language}
\label{subsec:basic-symb-op-sem}

The language used to specify a cryptographic protocol has the standard constructions, such as the creation of fresh values, sending and receiving messages. Moreover, it also includes ``if then else'' constructors needed to specify, for example, the RFID protocol used by passports. A protocol is composed of a set of roles.

\begin{definition} [Timed Protocols]
The set of Timed Protocols, $\TL$, is composed of Timed Protocol Roles, $\tl$, which are constructed by using commands as specified by the following grammar:
\[
  \begin{array}{l@{\qquad}l}
     \mathsf{nil} & \textrm{Empty Protocol}\\
     \mid (\new~ \v~\#~ \tc), \tl & \textrm{Fresh Constant}\\
     \mid (+ \m ~\#~ \tc), \tl &  \textrm{Timed Message Output}\\
     \mid (- \m ~\#~ \tc), \tl & \textrm{Timed Message Input}\\
     \mid (\mathsf{if}~{(\m_1 := \m_2)~\#~\tc}& \textrm{Timed Conditional}\\
     ~~\mathsf{then}~{\tl_1}~\mathsf{else}~{\tl_2}) 
  \end{array}
\]
\end{definition}
Intuitively, $\new$ generates a fresh value binding it to the variable $\v$, $(+ \m~\#~ \tc)$ denotes sending the term $\m$ and the $(- \m~\#~ \tc)$ receiving a term, and $(\myif{\m_1 := \m_2~\#~ \tc}{\mathsf{pl}_1}{\mathsf{pl}_2})$ denotes that if $\m_1$ can be matched with $\ms_2$, that is, instantiate the variables in $\m_1$ so that the resulting term is $\ms_2$, then the protocol proceeds by execution $\mathsf{pl}_1$ and otherwise to $\mathsf{pl}_2$. A command is only applicable if the associated constraint $\tc$ is satisfiable. We elide the associated time constraint whenever $\tc$ is a tautology, that is, it is always true.

\begin{example}
\label{ex:ns}
The Needham-Schroeder~\cite{ns} protocol is specified as follows where $X,Y$ are variables:
\[
\begin{array}{l}
 Alice :=  (\new ~N_a), (+ \enc(\tup{N_a,alice},\pk(Z))),
 \\ \qquad (- \enc(\tup{N_a,Y},\pk(alice)\})), (+ \enc(Y,\pk(Z)))\\[2pt]
 Bob :=  (- \enc(\tup{X,Z},\pk(bob))), (\new ~N_b),\\
 \qquad  (+ \enc(\tup{X,N_b},\pk(Z)\})), (- \enc(N_b,\pk(bob)))  
\end{array}
\]
\end{example}

\begin{example}
\label{ex:ns2}
  Consider the following protocol role which is a modification of Alice's role in the Needham-Schroeder's protocol (Example~\ref{ex:ns}):
\[
\begin{array}{ll}
 Alice := & (\new ~N_a),  (+\enc(\tup{N_a,alice},\pk(Z))), (- \v),\\ 
 & \quad \mathsf{if}~\v := \enc(\tup{N_a,Y},\pk(alice)\}) \\
& \quad\mathsf{then}~ (+ \enc(Y,\pk(Z)))\\
& \quad\mathsf{else}~(+ error)
\\[2pt]
\end{array}
\]
Here, Alice checks whether the received message $\v$ has the expected shape before proceeding. If it does not have this shape, then she sends an error message.
\end{example}

\begin{example}
\label{ex:distance-bounding}
  The following role specifies the verifier of a (very simple) distance bounding protocol~\cite{brands93eurocrypt}:
  \[
  \begin{array}{l}
   (\new~\v),(+ \v~\#~\tv = \cur), (- \v~\#~\cur \leq \tv + 4) 
  \end{array}    
  \]
  It creates a fresh constant and sends it to the prover, 
  remembering the current global time by assigning it to the time variable $\tv$. Finally, when it receives the response $\v$ it checks whether the current time is less than $\tv + 4$.
\end{example}

\begin{example}[Passport]
\label{ex:passport}
Timed conditionals can be used to specify the duration of operations, such as checking whether some message is of a given form. In practice, the duration of these operations can be measured empirically to obtain a finer analysis of the protocol~\cite{chothia10fc}.

For example, consider the following protocol role: 
\[
\begin{array}{l}
  (\new~\v),(+ \v),(- \{\v_{enc},\v_{mac}\}~\#~\tv_0 = \cur),\\
  \key{if}~ (\v_{mac} := \enc(\v_{enc},\k_M))~ \# ~ \tv_1 = \tv_0 + \tv_{Mac}\\
  \key{then}~ (\key{if}~(\v_{enc} := \enc(\v,\k_E)) ~\# ~ \tv_2 = \tv_1 + \tv_{Enc})\\
  \qquad ~\key{then} ~ (+ done ~\#~ \cur = \tv_2)~\key{else}~ (+ error ~\#~ \cur = \tv_2))\\
  \key{else}~(+ error ~\#~ \cur = \tv_1)
\end{array}  
\]
This role creates a fresh value $\v$ and sends it. Then it is expecting a pair of two messages $\v_{mac}$ and $\v_{enc}$, remembering at time variable $\tv_0$ when this message is received. It then checks whether the first component $\v_{mac}$ is of the form $\enc(\v_{enc},\k_M))$, \ie, it is the correct MAC. This operation  takes $\tv_{mac}$ time units. The time variable $\tv_1$ is equal to the time $\tv_0 + \tv_{mac}$, \ie, the time when the message was received plus the MAC check duration. If the MAC is not correct, an $error$ message is sent exactly at time $\tv_1$. Otherwise, if the first component, $\v_{MAC}$, is as expected, the role checks whether the second component, $\v_{enc}$, is an encryption of the form $\enc(\v,\k_E))$, which takes (a longer) time $\tv_{enc}$. If so it sends the $done$ message, otherwise the $error$ message, both at time $\tv_2$ which is $\tv_1 + \tv_{enc}$.
\end{example}

\begin{example}[Red Pill Example]
\label{ex:red-pill}
  We abstract the part of sending the baseline message, \eg, the messages that establish the connection to the server, and the part that sends the differential messages. We assume that it takes \key{dBase} to complete the exchange of the baseline messages.
\[
  \begin{array}{l}
(-(\key{baseline\_req})~ \#~ \tv_0 = \cur),\\
(+(\key{baseline\_done})~ \#~ \cur = \tv_0 + \key{dBase}),\\
(-(\key{diff\_req})~ \#~ \tv_1 = \cur )\\
(+ (\key{diff\_done})~ \#~ \cur = \tv_1 + \key{dAppl})
  \end{array}
\]
Then the part of the protocol that depends on the application starts. We abstract this part using the messages \key{diff\_req} and \key{diff\_done}. If the application is running over a virtual machine, then \key{dAppl} takes \key{dVirtual} time units; otherwise \key{dAppl} takes \key{dReal} time units, where \texttt{dVirtual > dReal}. 

The intruder can distinguish whether an application is running over a virtual machine or not by measuring the time it takes to complete the exchange of \key{diff\_req} and \key{diff\_done} messages. 
\end{example}

\begin{example}[Anonymous Protocol]
\label{ex:anonymous}
  We specify (a simplified version of the) anonymous group protocol proposed by Abadi and Fournet for private authentication~\cite{abadi04tcs}. Whenever a broadcasted message is received by an agent, it checks whether it has been encrypted with the group key $KB$. If this is the case, then it checks whether the player sending the message with key $\v_G$ is part of the group. If so, then it sends a response encrypted with his private key. Otherwise, he sends a decoy message. 
\[
  \begin{array}{l}
  -(\{\key{hello}), \enc(\{\key{hello},\v_n,\v_G,\},\k_G) \} ~\#~ \tv_0 = \cur\\
  \key{if}~\k_G: = KB ~\#~ \tv_1 = \tv_0 + \key{dEnc}\\
  \key{then} \\
  \quad \key{if}~\v_G: = KA ~\#~ \tv_2 = \tv_1 + \key{dCheck}\\
  \quad \key{then}~+(\{\key{ack}, \enc(\key{rsp},\k_B)\}) ~\#~ \cur = \tv_1 + \key{dCreate}\\
  \quad \key{else}~+(\{\key{ack}, \enc(\key{decoy},\k_B)\}) ~\#~ \cur = \tv_1\\
  \key{else}~+(\{\key{ack}, \enc(\key{decoy},\k_B)\}) ~\#~ \cur = \tv_1\\
  \end{array}
\]
Notice the use of time constraints to capture that the steps of the protocol take some time, namely \key{dCheck} and \key{dCreate}. 
\end{example}

\subsection{Operational Semantics for Timed Protocols}


The operational semantics of timed protocols is given in Figure~\ref{fig:os-branching}. The rewrite rules are rewrite configurations defined below:

\begin{definition}
\label{def:basic-configuration}
A symbolic term configuration has the  form $\tup{\Pscr, \IK, \DC, \EQ, \TC}@\tG$, where 
\begin{itemize}
  \item $\Pscr$ is a set of player roles of the form $[n \mid \pl \mid \ks]$ composed by an identifier, $n$, a protocol $\pl$, and a set of known keys $\ks$;
  \item $\IK$ is the intruder knowledge;
  \item $\DC$ is a set of derivability constraints;
  \item $\EQ$ is a set of comparison constraints;
  \item $\TC$ is a set of time constraints;
  \item $\tG$ is a time symbol representing global time.
\end{itemize}
\end{definition}

The operational semantics of timed protocols is defined in Figure~\ref{fig:os-timed}. The \textbf{New} rule replaces the (bound) variable $\v$ by a fresh nonce $\n^\nu$. The \textbf{Send} rule sends a message $\ms$ which is then added to the intruder knowledge. The \textbf{Receive} rule expects a term of the form $\m$. The function $\sgen(\m, \IK, \DC)$ returns the variable substitution $\sb$ and a set of solutions $\{\ssb,\DC_1\}~ \css$. Each solution intuitively generates a different trace. We apply $\sb$ in the remaining of the program $\pl$ and apply the symbol substitution $\ssb$ to all symbols in the resulting configuration. This rule also has a proviso that the message $\ms = \ssb[\sb[\m]]$ is encrypted with keys that can be decrypted by the honest participant. This is specified by the function $isReceivable$. Finally, it also adds to the set of keys of the honest participant $\ks$, the keys he can learn from the message $\ms$. The rule \textbf{If-true} checks whether the terms $\m_1$ and $\ms_1$ can be matched from the intruder knowledge $\IK$. This is done by the function $\sgenB$ which is defined in a similar fashion as $\sgen$. It then adds the equality constraint to the set of comparison constraints. Finally, the rule \textbf{If-false} replaces the variables in $\m$ by fresh symbols, constrained in $\DC'$ with the intruder knowlegde. That is if $\sym^\nu$ is a fresh symbol, then $\dc(\sym^\nu,\IK)$. It also adds the corresponding inequality constraint. The intuition of replacing variables in $\m_1$ by fresh symbols is to specify that for any instance of these variables, the resulting term cannot be matched with $\ms_2$ as specifies the inequality constraint.

\begin{example}
\label{ex:lowe}
Consider the Needham-Schroeder protocol in Example~\ref{ex:ns}. Assume that the intruder initially only knows his secret key (and the guessables), $\IK_0 = \{\sk(eve)\}$ and there are no symbols $\DC = \emptyset$. An execution of Alice's protocol role is as follows. 
Alice creates a fresh constant $N_a$ and sends the message $\enc(\tup{N_a,alice},\pk(eve))$. At this point, the intruder knowledge is:
\[
  \IK_1 = \IK_0 \cup \{N_a\}
\]
He now can send a message to $Bob$, namely $\enc(\tup{\sym_1,\sym_2},\pk(bob))$ where $\sym_1,\sym_2$ are fresh and constrained $\DC_1 = \{\dc(\sym_1, \IK_1), \dc(\sym_2, \IK_1)\}$. At this point, Bob creates a fresh value $N_b$ and sends the message $\enc(\tup{\sym_1,N_b},\pk(\sym_2)\})$. The intruder learns this message (and no further):
\[
  \IK_2 = \IK_1 \cup \{\enc(\tup{\sym_1,N_b},\pk(\sym_2)\})\}
\]
Now, the intruder can fool alice by sending her a message of the form $\enc(\tup{N_a,Y},\pk(alice)\})$. We create a fresh symbol $\sym_3$ for $Y$ obtaining $\enc(\tup{N_a,\sym_3},\pk(alice)\})$ and attempt to generate this message from $\IK_2$ using $\sgen$. Indeed we can generate this message using $\enc(\tup{\sym_1,N_b},\pk(\sym_2)\}) \in \IK_2$. This generates the $\ssb = [\sym_1 \mapsto N_a, \sym_2 \mapsto alice, \sym_3 \mapsto N_b]$. This substitution is consistent with $\DC_1$. Notice that $\sym_3$ is not constrained. The protocol finishes by the intruder simply forwarding the message send by alice to bob. Bob then thinks he is communicating with alice, but he is not.  
\end{example}

Each rule has two general provisos. The first is that the resulting set of comparison constraints should be consistent. This can be checked as defined in Definition~\ref{def:eqCheck}. 

The second, more interesting, condition is on the time symbols. Whenever a rule is applied, time constraints $\TC_1$ are added to the configuration's constraint set. These time constraints are obtained by replacing $\cur$ in $\tc$ with $\tG_1$ together with the constraint $\tG_1 \geq \tG_0$ specifying that time can only advance. The rule is fired only if the resulting set of time constraints ($\TC \cup \TC_1$) is consistent, which can be done using SMT solver. This way of specifying systems is called Rewriting Modulo SMT~\cite{rocha-2012}.

\begin{figure*}[t]
  \[
  \begin{array}{l}
  \textbf{New:~}\tup{[n\mid (\new~ \v~\#~ \tc), \mathsf{pl} \mid \ks]~\Pscr,\IK,\DC,\EQ,\TC}@\tG_0 \lra   \tup{[n\mid \sb[\mathsf{pl}] \mid \ks]~ \Pscr,\IK,\DC,\EQ,\TC_1}@\tG_1\\
\textrm{where $\n^\nu$ is a fresh nonce and $\sb = [\v \mapsto \n^\nu]$}\\[3pt]

    \textbf{Send:~}\tup{[n\mid (+ \ms ~\#~ \tc), \mathsf{pl} \mid \ks]~\Pscr, \IK, \DC, \EQ, \TC}@\tG_0 \lra  \tup{[n\mid \mathsf{pl} \mid \ks]~ \Pscr, \IK \cup\{\ms\}, \DC, \EQ, \TC_1}@\tG_1\\[3pt]

\textbf{Receive:~}\tup{[n\mid (-\m ~\#~ \tc), \mathsf{pl} \mid \ks]~\Pscr, 
      \IK, \DC, \EQ, \TC}@\tG_0 \lra \\
      \qquad \qquad \qquad \ssb[\tup{[n\mid \sb[\mathsf{pl}] \mid \addKeys(\ms,\ks)]~ \Pscr, \IK, \DC_1, \EQ,\TC_1}]@\tG_1\\[1pt]
    \textrm{where~} 
    \{\sb,\{\ssb,\DC_1\}~ \css \} := \sgen(\m,\IK,\DC)\textrm{ and }\ms = \ssb[\sb[\m]]\textrm{ and } isReceivable(\ms,\ks) \\[3pt]

\textbf{If-true:~}
   \tup{[n\mid (\myif{(\m_1 := \ms_2~\#~ \tc) }{\mathsf{pl}_1}{\mathsf{pl}_2}) \mid \ks]~\Pscr,\IK, \DC, \EQ,\TC}@\tG_0 \lra \\
      \qquad \qquad \qquad 
      \ssb[\tup{[n\mid \sb[\mathsf{pl_1}] \mid \ks~ \Pscr,\IK, \DC_1, \EQ \cup \{\Eq(\sb[\m_1], \sb[\ms_2])\},\TC_1}]@\tG_1\\[1pt]
    \textrm{where~} 
    \{\sb,\{\ssb,\DC_1\}~ \css \} := \sgenB(\m_1 = \ms_2,\IK,\DC)\\[3pt]

\textbf{If-false:~}
   \tup{[n\mid (\myif{\m_1 := \ms_2}{\mathsf{pl}_1}{\mathsf{pl}_2}) \mid \ks]~\Pscr, \IK, \DC, \EQ} \lra \\
      \qquad \qquad \tup{[n\mid \sb[\mathsf{pl_2}] \mid \ks]~\Pscr,\IK, \DC \cup \DC', \EQ \cup\{\Neq(\sb[\m_1],\ms_2)\}} \\[1pt]
          \textrm{where~$\sb$ replaces the variables in $\m_1$ by fresh symbols which are constrained with $\IK$ in $\DC'$}\\[1pt]

  \end{array}
  \]
  \caption{Operational semantics for basic protocols. Here $\sb$ is a substitution mapping the variables in $\m$ by fresh symbols; and the function $rng$ applies the symbol substitution $\ssb$ to the range of the variable substitution $\sb$; $\tc_1$ is the time constraint obtained by replacing $\cur$ in $\tc$ by the global time $\tG_1$; and $\TC_1 = \TC\cup\{\tG_1 \geq \tG_0,\tc_1\}$. The function isReceivable checks whether the message $\ssb[\sb[\m]]$ can be decrypted with the keys he has in $\ks$. Every rule has the proviso that the set of comparison constraints and the set of time constraints should be satisfiable.
  Rules are only applicable if the set of time constraints are consistent.}
  \label{fig:os-timed}
  \label{fig:os-basic}
  \label{fig:os-branching}
\end{figure*}

\begin{definition}
\label{def:basic-traces}
Let $\Rscr$ be the set of rules in Figure~\ref{fig:os-timed}. 
  A timed trace is a labeled sequence of transitions written $\Cscr_1 \stackrel{l_1}{\lra} \Cscr_2 \stackrel{l_2}{\lra} \cdots \stackrel{l_{n-1}}{\lra} \Cscr_n$ such that for all $1 \leq i \leq n-1$, $\Cscr_i \lra \Cscr_{i+1}$ is an instance of a rule in $\Rscr$ and $l_i$ is $+\ms @ \tG_1$ if it is an instance of Send rule sending term $\ms$ at time $\tG_1$, $-\ms @ \tG_1$ if it is an instance of Receive rule receiving term $\ms$  at time $\tG_1$, and $\emptyset$ otherwise.
\end{definition}

The use of rewriting modulo SMT considerably reduces the search space.  Timed protocols are infinite state systems, as time symbols can be instantiated by any (positive) real number. With the use of rewriting modulo SMT we simply have to accumulate constraints. Only traces with satisfiable sets of time constraints are allowed. Indeed, as we describe in Section~\ref{sec:exp}, the number of traces is not only finite (as stated in the following Proposition), but very low (less than 40 traces). As observational equivalence involves the matching of traces, checking for observational equivalence can be automated.

\begin{proposition}
\label{prop:finite-traces}
  The set of traces starting from any configuration $\Cscr_0$ is finite.
\end{proposition}

\begin{proposition}
\label{prop:trace-properties}
  Let $\tau = \Cscr_1 \stackrel{l_1}{\lra} \Cscr_2 \stackrel{l_2}{\lra} \cdots \stackrel{l_{n-1}}{\lra} \Cscr_n$ be a trace. For any $\Cscr_i = \tup{\Pscr_i, \IK_i, \DC_i, \EQ_i}$, such that $1 \leq i \leq n$, the following holds:
  \begin{itemize}
     \item For any $i \leq j \leq n$, $ \DC_n(\IK_{i})|_{\EQ_n} \subseteq \DC_n(\IK_{j})|_{\EQ_n}$, that is, the intruder knowledge can only increase;
     \item $\DC_{i}$ is an acyclic set of derivability constraints; 
     \item Let $\sym_k$ be a symbol created in some $\Cscr_k$, $k < i$ and let $\sym_i$ be a symbol created in $\Cscr_i$. If $\dc(\sym_k,\Sscr_k),\dc(\sym_i,\Sscr_i) \in \DC_i$, then $\Sscr_k \subseteq \Sscr_i$. That is, symbols that are introduced at a later transitions can be instantiated by more terms than symbols introduced at earlier transitions.
   \end{itemize} 
\end{proposition}

\paragraph{Timed Intruders:} In fact, our implementation generalizes the machinery in this section by considering multiple timed intruders~\cite{kanovich14fccfcs,nigam16esorics}. As described in~\cite{kanovich14fccfcs}, the standard Dolev-Yao may not be suitable for the verificaiton of Cyber-Physical Security Protocols where the physical properties of the environment is important. Differently from the Dolev-Yao intruder, a timed intruder needs to wait for the message to arrive before he can learn it. \cite{nigam16esorics} proved an upper-bound on the number of timed intruders. Our tool implements this strategy. However, for the examples considered here, the standard Dolev-Yao intruder is enough.

%% file: symbolicCC.tex
\subsection{Symbolic Constraint Solving}
\label{subsec:basic-solving}
For protocol verification, we will assume a traditional Dolev-Yao intruder~\cite{DY}, that is, an intruder that can construct messages from his knowledge by tupling and encrypting messages. However, he cannot decrypt a message for which he does not possess the inverse key. This is captured by the definition of minimal sets Definition~\ref{def:minimal}.

\begin{definition}
  An intruder knowledge $\IK$ is a minimal set of symbolic terms. 
\end{definition}

During protocol execution, the intruder sends messages to  honest participants constructed from his knowledge base. 
Suppose an honest player is ready to receive a message
matching a term $\m$, possibly containing variables.
Rather than considering all possible ground instances of $\m$ that the intruder could send, we consider
a finite representation of thie set, namely symbolic messages where the possible values of the symbols are constrained  by derivability constraints.  
To compute this this representation the intruder replaces variables with symbolic terms, possibly containing fresh symbols, and then constrains the symbols so that the
allowed instances are exactly the terms matching
$\m$ that the intruder can derive from his current
knowledge $\IK$. 

For example, consider the term $\m = \enc(\{\v_1,\sym,\v_1,\v_2\},\k)$ (which is expected as input by an honest player). Here $\v_1$ and $\v_2$ are variables and $\sym$ is constrained by derivability constraints $\DC$. We create two fresh symbols $\sym_1$ and $\sym_2$ for, respectively, the variables $\v_1$ and $\v_2$. We use  $\sb$  to denote such substitution of variables by symbolic terms. In this example $\sb = [\v_1 \mapsto \sym_1, \v_2 \mapsto \sym_2]$. We then obtain $\ms = \sb[\m] = \enc(\{\sym_1,\sym,\sym_1,\sym_2\},\k)$.


It remains to solve the following problem:

\begin{center}
\emph{Given an intruder knowledge, $\IK$, and a set of derivability constraints $\DC$ constraining the symbols in $\IK$, find
a representation of all instances of a symbolic term $\ms$, satisfying $\DC$, that can be generated from $\IK$.} 
\end{center}

We implemented the function called $\sgen$ that enumerates all possible instances. Its specification is in the Appendix. We describe $\sgen$ informally next and illustrate it with some examples. A similar algorithm is also used by~\cite{cortier09csf}.

In particular, $\sgen(\m, \IK,\DC)$ takes as input a term $\m$, which is expected by the honest participant, the intruder knowledge $\IK$ and the derivability constraints $\DC$ for the existing symbols. $\sgen(\m, \IK,\DC)$ then generates as output a pair:
\[
  \{\sb, \{\ssb_1, \DC_1\} \ldots \{\ssb_k, \DC_k\}\}
\]
where $\sb$ maps the variables of $\m$ to symbols, and each $\{\ssb_i,\DC_i\}$ is a solution to the problem above for $\ms = \sb[\m]$. If $k = 0$, then there are no solutions, that is, the intruder is not able to generate a term which matches $\m$. 

Intuitively, the function $\sgen$ constructs a solution by either matching $\m$ with a term in his knowledge $\IK$ (base case) or constructing $\m$ from terms in $\IK$ and using tupling and encryption. The following examples illustrates the different cases involved:

\begin{example}
Consider the following cases for deriving the term $\m = \enc(\{\v,\sym\},\k)$. 
\begin{itemize}
  \item Case 1 (matching with a term in $\IK$): Assume:
  \[
\begin{array}{l}
  \IK = \{\enc(na,\sym_1),\k)\}\\
  \DC = \dc(\sym, \{\n_a,\n_c\})\ \dc(\sym_1, \Sscr)
\end{array}
\]
Then the solution of $\sgen$ is: 
\[
 \{\sb, \{[\sym_\v \mapsto \n_a, \sym \mapsto \sym_1], \dc(\sym_1, \add(\{na,nc\},\Sscr) )\}\} 
\]
where $\sb = [\v \mapsto \sym_\v]$ and $\sym_\v$ is a fresh symbol. Notice that since $\sym_\v$ is mapped to a particular term ($\n_a$), no derivability constraint for it is generated. Additionally, notice that $\sym$ is constrained to be the same as $\sym_1$. This causes the removal of the derivability constraint $\dc(\sym_1, \Sscr)$;

  \item Case 2 (constructing terms from $\IK$): Assume that $\k \in \IK$ and $\IK$ has no encryption term. Then the solution of $\sgen$ is:
\[
  \{[\v \mapsto \sym_\v], \{[], \DC\}\}
\]  
which corresponds to generatign the term $\enc(\{\sym_v,\sym\},\k)$.

\item Case 3 [No Solution]: Assume that  $\IK = \{\enc(\n_a,\n_b),\k)\}$ and $\DC = \dc(\sym, \{na,nc\})$. Since $\sym$ cannot be instantiated to $\n_b$, the intruder cannot use the term $\enc(\n_a,\n_b),\k)$.
\end{itemize}




\end{example}

%% file: sec-obs.tex
Our goal now is to determine when two term configurations $\Cscr_I = \tup{\Pscr_I,\IK_I,\DC_I,\EQ_I,\TC_I}@\tG$ and $\Cscr_I' = \tup{\Pscr_I',\IK_I',\DC_I',\EQ_I',\TC_I'}@\tG'$ cannot be distinguished by the Dolev-Yao intruder. That is, for any trace starting from $\Cscr_I$ there is \emph{an equivalent trace} starting from $\Cscr_I'$. The following definition specifies observables which collect the necessary information from a trace:

\begin{definition}
\label{def:basic-observables}
  Let $\tau = \Cscr_1 \stackrel{l_1}{\lra} \Cscr_2 \stackrel{l_2}{\lra} \cdots \stackrel{l_{n-1}}{\lra} \Cscr_n = \tup{\Pscr_n,\IK_n,\DC_n,\EQ_n,\TC_n}@\tG_n$ be a timed trace. Its observable is the tuple $\tup{\tv_I,\Lscr_\tau,\IK_n,\DC_n,\EQ_n,\TC_n}$, where $\tv_I$ is the global time at configuration $\Cscr_1$, $\Lscr_\tau$ is the sequence of non-empty labels in $\tau$.
  Let $\Cscr$ be a configuration. Let $\Tscr(\Cscr)$ be the set of all traces with initial configuration $\Cscr$. The  observables of $\Cscr$ is $\Oscr(\Cscr) = \{\Oscr_\tau \mid \tau \in \Tscr(\Cscr)\}$, that is, the set of all observables of traces starting from $\Cscr$.
\end{definition}




Two configurations are observationally equivalent if their observables are equivalent.

 \begin{definition}
   A configuration $\Cscr$ approximates a configuration $\Cscr'$, written $\Cscr \preceq \Cscr'$ if for any $\Oscr \in \Oscr(\Cscr)$ there exists an equivalent observable $\Oscr'\in \Oscr(\Cscr')$, that is, $\Oscr \sim \Oscr'$ (Definition~\ref{def:equiv-obs}). The configurations are observationally equivalent, written $\Cscr \sim \Cscr'$, if and only if $\Cscr \preceq \Cscr'$ and $\Cscr' \preceq \Cscr$.
 \end{definition}

\begin{definition}
\label{def:equiv-obs}
   Let $\Oscr = \tup{\Lscr, \IK, \DC, \EQ, \TC}@\tG$ and $\Oscr' = \tup{\Lscr', \IK', \DC', \EQ', \TC'}@\tG'$ be two observables, such that $\Lscr = \{(\pm_1 \ms_1@\tG_1) \ldots (\pm_p \ms_p@\tG_p)\}$ and $\Lscr' =\{(\pm_1' \ms_1'@\tG_1') \ldots (\pm_n' \ms_n'@\tG_n')\}$. The observation $\Oscr$ is equivalent to $\Oscr'$, written $\Oscr \sim \Oscr'$ if the following conditions are all true:
\begin{enumerate} 
  \item $p = n = N$, that is, they have the same length $N$;
  \item $\pm_i = \pm_i'$, for all $1 \leq i \leq N$, that is, have the same label type;
  
  \item \label{defitem:blackbox} There is a black-box bijection $\bij(\Oscr, \Oscr')$ between $\Oscr$ and $\Oscr'$ (Definition~\ref{def:black-box-bij});

  \item \label{defitem:termapprox} The messages observed are equivalent, that is, $\tup{\ms_1, \ldots, \ms_N} \sim_{\Oscr,\Oscr'} \tup{\ms_1', \ldots, \ms_N'}$ (Definition~\ref{def:term-approx-denotational});

  \item \label{defitem:timeapprox} Assume $\widetilde{\tv}$ and $\widetilde{\tv'}$ are the set of time symbols in $\TC$ and $\TC'$, respectively. The following formulas are tautologies:
  \[
  \begin{array}{l}
      \forall \widetilde{\tv}. \left[\TC \Rightarrow \exists \widetilde{\tv'}. \left[\TC' \land 
      \tG_1 = \tG_1' \land \cdots \land \tG_N = \tG_N' \right]\right]\\[5pt]
      \forall \widetilde{\tv'}. \left[\TC' \Rightarrow \exists \widetilde{\tv}. \left[\TC \land 
      \tG_1 = \tG_1' \land \cdots \land \tG_N = \tG_N' \right]\right]
  \end{array}
    \]
    specifying that for any time $\tG_1, \ldots, \tG_N$ satisfying $\TC$ when the messages $\ms_1, \ldots, \ms_N$ were observed in $\Oscr$, we can find $\tG_1', \ldots, \tG_N'$ satisfying $\TC'$, that makes the time of observation equal. 
\end{enumerate}
\end{definition}

The first two conditions are clear. If two observables differ on the number of observations or they differ on their types, then they can be distinguished. We motivate and define next the conditions~\ref{defitem:blackbox} and \ref{defitem:termapprox} of black-box bijections and term approximation. Finally, the condition \ref{defitem:timeapprox} specifies that the intruder cannot distinguish the observed messages by measuring the time when they are observed.

\subsection{Black-Boxed Terms}
There are some subtleties involved in handling nonces and handling encrypted terms for which the intruder does not possess the inverse key. Consider the two following sequence of labels:
\[
  \begin{array}{l}
    \Lscr_1 = \tup{(+ \n_1),(+ \n_2),(+ \tup{\n_1,\n_2}) }\\
    \Lscr_2 = \tup{(+ \n_1),(+ \n_2),(+ \tup{\n_2,\n_1}) }
  \end{array}
\]
where $\n_1, \n_2$ are nonces, say created by an honest participant. The intruder can distinguish these observables because of the order of the pair of the last send messages. However, the intruder should not be able to distinguish the following two sequence of labels:
\[
  \begin{array}{l}
    \Lscr_1' = \tup{(+ \n_1),(+ \n_2),(+ \tup{\n_1,\n_2}) }\\
    \Lscr_2' = \tup{(+ \n_1'),(+ \n_2'),(+ \tup{\n_1',\n_2'}) }
  \end{array}
\]
where $\n_1,\n_1',\n_2,\n_2'$ are nonces. This is because although different constants, they are appear in the same order. The same happens for encrypted terms that the intruder cannot decrypt.\footnote{Recall that from the Dolev-Yao intruder's point of view, such terms are treated as black-boxes.} Consider the following two sequence of labels:
\[
  \begin{array}{l}
    \Lscr_1'' = \tup{(+ \enc(\t_1,\k_1)),(+ \enc(\t_2,\k_2)),(+ \tup{\enc(\t_1,\k_1),\enc(\t_2,\k_2)}) }\\
    \Lscr_2'' = \tup{(+ \enc(\t_1',\k_1')),(+ \enc(\t_2',\k_2')),(+ \tup{\enc(\t_2',\k_2'),\enc(\t_1',\k_1')}) }
  \end{array}
\]
where $\t_1, \t_2$ are different constants and $\k_1, \k_2$ are keys for which the intruder does not possess the decryption key. Although the intruder cannot decrypt the exchanged messages, he can still distinguish $\Lscr_1''$ and $\Lscr_2''$ because of the order of last tuple sent. 



The definition below specifies when two observables are equal. It uses definition of black-box bijections and term equivalences defined in the following.

\begin{definition}
  Let $\Oscr = \tup{\Lscr, \IK, \DC, \EQ, \TC}@\tG$ be an observable. 
  The set of black-box terms of $\Oscr$ is the set, written $\BB(\Oscr)$, of all sub-terms in $\Lscr$ that are nonces, $\n$, or encryption terms, $\enc(\t,\k)$, for which the inverse of $\k$ is not in $\IK$.
\end{definition}

\begin{definition}
 Let $\Oscr$ be an observable. A symbolic term $\ms$ restricted to $\BB(\Oscr)$, written, $\ms|_{\BB(\Oscr)}$, is the term obtained by replacing  by the special symbol $*$ all sub-terms of $\ms$ that (1) do not contain any term in $\BB(\Oscr)$ or (2) is not contained in a term in $\BB(\Oscr)$. 
 \end{definition}

For example, the consider the term $\ms = \tup{\n_1,\t_1,\enc(\tup{\t_2,\enc(\t_3,\k_2)},\k_1)}$ where the intruder knows the inverse of $\k_1$, but not of $\k_2$. Then the term restricted to black-boxes is $\tup{\n_1,*,\enc(\tup{*,\enc(\t_3,\k_2)},*)}$, where $\t_1, \t_2$ and $\k_1$ are replaced by $*$.

\begin{definition}
\label{def:black-box-bij}
   Let $\Oscr = \tup{\Lscr, \IK, \DC, \EQ, \TC}@\tG$ and $\Oscr' = \tup{\Lscr', \IK', \DC', \EQ', \TC'}@\tG'$ be two observables, such that $\Lscr = \{(\pm_1 \ms_1) \ldots (\pm_n \ms_n)\}$ and $\Lscr' =\{(\pm_1 \ms_1') \ldots (\pm_n \ms_n')\}$.  A black-box bijection, $\bij(\Oscr,\Oscr')$ is any bijection between the sets $\BB(\Oscr)$ and $\BB(\Oscr')$ such that $\bij(\Oscr,\Oscr')$ makes the terms $\tup{\ms_1,\ldots,\ms_n}|_{\BB(\Oscr)}$ and  $\tup{\ms_1',\ldots,\ms_n'}|_{\BB(\Oscr)}$ equal.\footnote{
One can easily compute one such bijection by recursively traversing terms and checking which ones are black-boxed or not. This procedure is implemented in the file obs-equiv.maude, function \texttt{mkSSBB}.}
\end{definition}

For example, there is no bijection between the terms in $\Lscr_1''$ and $\Lscr_2''$ shown above. However, the bijection $\{\enc(\t_1,\k_1) \leftrightarrow \enc(\t_1',\k_1'), 
  \enc(\t_2,\k_2) \leftrightarrow \enc(\t_2',\k_2')\}$
makes equal the terms in $\Lscr_1''$ and the sequence of labels $\Lscr_2'''$ below:
\[
    \Lscr_2''' = \tup{(+ \enc(\t_1',\k_1')),(+ \enc(\t_2',\k_2')),(+ \tup{\enc(\t_1',\k_1'),\enc(\t_2',\k_2')}) }.
\]

\subsection{Term approximation}



\textbf{Notation:} For the remainder of this section, to avoid repetition,  we will assume given two observables $\Oscr = \tup{\Lscr,\IK,\DC, \EQ, \TC}@\tG$ and $\Oscr' = \tup{\Lscr',\IK',\DC',\EQ', \TC'}@\tG'$ with disjoint sets of term and time symbols and of nonces. Moreover, $\bij = \bij(\Oscr,\Oscr')$ is a bijection between the black-box terms of the observables $\Oscr$ and $\Oscr'$.

First we handle black boxed in the same way as nonces: given a bijection with $n$ relations:
\[
 \bij =  \{ \ms_1 \leftrightarrow \ms_1', \ldots, \ms_n \leftrightarrow \ms_n'\}
\]
We create $n$ new nonces $\n^\nu_1, \ldots, \n^\nu_n$ and replace each occurrence of $\ms_i$ in $\Oscr$ and $\ms_i'$ in $\Oscr'$ by the same nonce $\n^\nu_i$ for $1 \leq i \leq n$. In the following we assume that this replacement has been already performed.

We define when a symbolic term approximates another one.

\begin{definition}
\label{def:term-approx-denotational}
 Let $\ms$ and $\ms'$ be two symbolic terms in, respectively, $\Oscr$ and $\Oscr'$. The term $\ms$ approximates $\ms'$, written $\ms \preceq_{\Oscr,\Oscr'} \ms'$ if for all ground terms $\m \in \DC\mid_\EQ(\ms)$, then $\m \in \DC'\mid_{\EQ'}(\ms')$. Two terms are equivalent, written $\ms \sim_{\Oscr,\Oscr'} \ms'$ if and only if $\ms \preceq_{\Oscr,\Oscr'} \ms'$ and $\ms' \preceq_{\Oscr,\Oscr'} \ms$.
\end{definition}

\newcommand\termApprox{\mathsf{termApprox}}
\newcommand\symDer{\mathsf{symDer}}

A problem with the definition above is the quantification over all terms, which is an infinite set. We define in Definition~\ref{def:term-eq-approx} a procedure, called $\termEqApprox$ which checks for the observational equivalence of terms. It uses the auxiliary functions $\termApprox$ and $\canEq$. 

We start by defining the function $\termApprox$ which uses the auxiliary function $\symDer$. 

\begin{definition}
\label{def:term-approx}
 Let $\ms$ and $\ms'$ be terms in, respectively, $\Oscr$ and $\Oscr'$. The predicate $\termApprox(\ms,\ms',\Oscr,\Oscr')$ evaluates to true if and only if the two conditions below are satisfied:
\begin{enumerate}
  \item 
  \label{defitem:term-approx-match}
  There exists a matching substitution $\theta = \{\sym'_1 \mapsto \ms_1, \ldots, \sym_n' \mapsto \ms_n\}$ mapping symbols in $\ms'$ to subterms in $\ms$, such that, $\theta[\ms'] = \ms$;
  \item 
  \label{defitem:term-approx-sym-approx}
  For each $\sym_i' \mapsto \ms_i \in \theta$, we have $\symDer(\sym_i',\ms_i,\DC,\DC')$ (Definition~\ref{def:sym-approx}).
\end{enumerate}
\end{definition}

\begin{definition}
\label{def:sym-approx}
Let $\sym'$ be a symbol in $\Oscr'$ and $\dc(\sym',\Sscr') \in \DC'$ and $\ms$ be a term in $\Oscr$. We say that $\symDer(\sym',\ms,\DC,\DC')$ if one of the following holds:
  \begin{itemize}
    \item $\ms$ is an guessable;
    \item $\ms$ is a key, then $\ms \in \Sscr'$;
    \item $\ms$ is a nonce $\n$, then $\n \in \Sscr'$;
    \item $\ms$ is an encryption $\enc(\ms_1,\k) \notin \BB(\Oscr)$, then $
    \symDer(\sym',\tup{\ms_1,\k},\DC,\DC')$;
    \item $\ms = \tup{\ms_1,\ldots,\ms_n}$, then 
    $\symDer(\sym',\ms_i,\DC,\DC')$ for $1 \leq i \leq n$;
    \item $\ms = \sym$ is a symbol such that $\dc(\sym,\Sscr) \in \DC$, then for each $\ms_1 \in \Sscr$, $\symDer(\sym',\ms_1,\DC,\DC')$.
  \end{itemize}
\end{definition}

We prove the following soundness and completeness results for the functions $\symDer$ and $\termApprox$:

\begin{theorem}
\label{th:basic-obs}
  Let $\ms$ and $\ms'$ be two terms. The for all $\m \in \DC(\ms)$, we have $\m \in \DC(\ms')$ if and only if $\termApprox(\ms,\ms',\DC,\DC')$. 
\end{theorem}

\paragraph{Example}: We illustrate the procedure $\termApprox$ with an example. Consider
\[
  \begin{array}{l}
  \ms = \tup{\n,\enc(\tup{\n,\t,\sym_1,\sym_2},\k_1)} \\
  \ms' = \tup{\n',\enc(\sym',\k_1)} \\
  \end{array}
\]
Moreover, assume that the intruder knows the inverse of $\k_1$ in both observables, that is, $\k_1^{-1} \in \IK \cap \IK'$, and consider the bijection $\bij = \bij(\Oscr,\Oscr') = \{\n \leftrightarrow \n'\}$. We replace these nonces by the fresh nonce $\n^\nu$.

We check for Condition~\ref{def:term-approx}.\ref{defitem:term-approx-match}, that is, construct a matching substitution mapping symbols in $\ms'$ to subterms in $\ms$. We obtain the following matching substitution:
\[
  \theta = \{\sym' \mapsto \tup{\n^\nu,\t,\sym_1,\sym_2}\}
\]
as $\theta[\ms'] = \tup{\n^\nu,\enc(\tup{\n',\t,\sym_1,\sym_2},\k_1)}$.

Notice that if such a matching substitution does not exists, then the $\ms$ cannot approximate $\ms'$. Indeed $\ms' \npreceq_{\Oscr,\Oscr'} \ms$ as there is no such matching substitution: The term $\ms'$ can be instantiated to a term where a non-tuple is encrypted, \eg, $\tup{\n_1,\enc(\t,\k_1)}$ for some guessable $\t$, whereas $\ms$ can only be instantiated by terms where a tuple with at least four elements is encrypted. 

Now we check for the Condition~\ref{def:term-approx}.\ref{defitem:term-approx-sym-approx}. In order to check this condition, we need to know more about $\DC$ and $\DC'$. Assume that the symbol $\sym_1$ is created before the symbol $\sym_2$ in the trace corresponding to $\Oscr$. Then by Proposition~\ref{prop:trace-properties}, we have that if $\dc(\sym_1, \Sscr_1), \dc(\sym_2, \Sscr_2) \in \DC$, then $\Sscr_1 \subseteq \Sscr_2$. 

Consider, $\dc(\sym_1, \{\t_1,\k_1\}), \dc(\sym_2, \{\t_1,\t_2,\k_1,\k_2\}) \in \DC$ which satisfies the condition above. Moreover, assume $\dc(\sym',\Sscr') \in \DC'$ with $\Sscr' = \{\n^\nu,\t_1,\k_1,\t_2,\k_2\}$. 

We have that $\symDer(\sym',\{\n, \sym_1, \sym_2\},\DC,\DC')$ as:
\begin{itemize}
  \item $\symDer(\sym',\n^\nu,\DC,\DC')$  as $\n^\nu \in \Sscr'$;
  \item $\symDer(\sym',\sym_1,\DC,\DC')$ as $\{\t_1,\k_1\} \subseteq \Sscr'$;
  \item $\symDer(\sym',\sym_2,\DC,\DC')$ as $\{\t_1,\t_2,\k_1,\k_2\} \subseteq \Sscr'$;
\end{itemize}

The following definition specifies the function $\canEq$.
Intuitively, $\canEq$ checks whether it is possible to instantiate w.r.t. $\DC,\EQ$ the symbols in $\ms_1,\ms_2$ so to falsify the constraint $\Neq(\ms_1,\ms_2)$. 

\begin{definition}
\label{def:neq-approx}
$\canEq(\ms_1,\ms_2,\DC,\EQ)$ evaluates to true if and only if the following conditions are satisfied:
\begin{enumerate}
  \item If is a unifier $\sigma = [\sym_1 \mapsto \ms_1^*, \ldots, \sym_n \mapsto \ms_n^*]$ of $\ms_1$ and $\ms_2$;
  \item $\sigma$ does not make any comparison constraint in $\EQ$ false;
   \item For each $\sym_i \mapsto \ms_i^* \in \sigma$, $\sym_i$ can generate the term $\ms_i^*$ w.r.t. $\DC$.
    \label{defitem:neq4} 
\end{enumerate}
\end{definition}

\begin{definition}
\label{def:term-eq-approx}
 Let $\ms$ and $\ms'$ be terms in, respectively, $\Oscr$ and $\Oscr'$. The predicate $\termEqApprox(\ms,\ms',\Oscr,\Oscr')$ evaluates to true if and only if:
\begin{enumerate}
  \item If $\eqCheck(\EQ,\DC) = false$;
  \label{def:item-eqcheck1}
  \item Otherwise if $\eqCheck(\EQ,\DC) = \eqCheck(\EQ',\DC') = true$, let $\ssb$ and $\ssb'$ being the corresponding witnessing  matching subsitutions (Definition~\ref{def:eqCheck});
  \label{def:item-eqcheck2}
  \begin{enumerate}
  \item
  \label{defitem:term-eq-approx}
  $\termApprox(\ssb[\ms],\ssb'[\ms'],\ssb[\DC],\ssb'[\DC'])$ is true where the witnessing matching subsitution is  $\theta = [\sym'_1 \mapsto \ms_1, \ldots, \sym_n' \mapsto \ms_n]$ (Definition~\ref{def:term-approx});
  \item 
  \label{defitem:neq}
   $\canEq(\ssb[\theta[\ms_1']],\ssb[\theta[\ms_2']],\DC,\EQ)$ is false (Definition~\ref{def:neq-approx}) for each inequality constraint $\Neq(\m_1',\m_2') \in \EQ'$.
  \end{enumerate}
\end{enumerate}
\end{definition}

Intuitively, the Condition \ref{def:term-eq-approx}.\ref{def:item-eqcheck1} specifies that if $\EQ, \DC$ are not consistent, then the approximation is trivial. Otherwise, if both $\EQ,\DC$ and $\EQ',\DC'$ are consistent, then there are matching substitutions $\ssb$ and $\ssb'$. We apply $\ssb$ and $\ssb'$ substitutions in, respectively, $\Oscr$ and $\Oscr'$ in the following. This takes care of all the $\Eq$ constraints. Condition \ref{def:term-eq-approx}.\ref{defitem:term-eq-approx} checks whether the terms they generate are the same.  Condition \ref{def:term-eq-approx}.\ref{defitem:neq} intuitively checks whether it is not  possible to falsify any inequality constraint in $\EQ'$ using $\DC$ and $\EQ$. 

\begin{theorem}
\label{th:branching-obs}
  Let $\ms$ and $\ms'$ be two terms from, respectively, $\Oscr$ and $\Oscr'$. Then $\ms \preceq_{\Oscr,\Oscr'} \ms'$ if and only if $\termEqApprox(\ms,\ms',\Oscr,\Oscr')$. 
\end{theorem}

\paragraph{Example}
Consider the protocol role for Alice in Example~\ref{ex:ns2}. We omit the time constraints as they are not important for this example. From an initial configuration with Alice and the intruder called $eve$, one can construct the following observable $\Oscr'$:
\[
  \begin{array}{l}
    \tup{+ \enc(\tup{\n_1',alice},\pk(eve)), - \sym', + error},\\ \IK', \DC', \{\Neq(\sym',\enc(\tup{n_1,\v},\pk(alice)\}) \}
  \end{array}
\]
for some intruder knowledge $\IK'$ and $\DC'$.

Consider the protocol role which outputs an error:
\[
\begin{array}{ll}
 Alice' := & (\new ~N_a),  (+\enc(\tup{N_a,alice},\pk(Z))), (- \v),(+ error)
\\[2pt]
\end{array}
\]
Consider the observable $\Oscr$ for this protocol:
\[
  \begin{array}{l}
    \tup{+ \enc(\tup{\n_1,alice},\pk(eve)), - \sym, + error},\IK, \DC, \emptyset) \}
  \end{array}
\]
which is similar to the observable above, but without the comparison constraint. There is the bijection $\n_1 \leftrightarrow \n_1'$. Moreover, since the message is sent to $eve$, $\n_1 \in \IK$ and $\n_1' \in \IK'$. 

Notice that%
\begin{small}
\[
\begin{array}{c}
  \termApprox(\tup{\enc(\tup{\n_1,alice},\pk(eve)),sym,error},\\
  \tup{\enc(\tup{\n_1',alice},\pk(eve)),sym',error},\DC, \DC')  
\end{array}
\]   
 \end{small}%
as the messages are the same. 
However, due to the comparison constraint in $\Oscr$, it is not the case that 
\begin{small}
\[
\begin{array}{c}  
  \termEqApprox(\tup{\enc(\tup{\n_1,alice},\pk(eve)),sym,error},
  \\\tup{\enc(\tup{\n_1',alice},\pk(eve)),sym',error}, \Oscr, \Oscr')
\end{array}
\]  
\end{small}%
Indeed, the function $\canEq(\sym, \enc(\tup{n_1,\t},\pk(alice)\}, \DC, \emptyset)$ is true. (Recall that variables are replaced by fresh symbols.) This is true because the intruder knows the public key of Alice and $\n_1$.

\subsection{Timing of Messages}
\label{subsec-timing}
For Condition~\ref{def:equiv-obs}.\ref{defitem:timeapprox}, we reduce the formulas to formulas for which existing solvers can be used~\cite{duterte15smt}, namely formulas of the form $\exists \forall$:

\noindent
\begin{small}
      \[
    \begin{array}{l}
      \forall \widetilde{\tv}. \left[\TC \Rightarrow \exists \widetilde{\tv'}. \left[\TC' \land 
      \tG_1 = \tG_1' \land \cdots \land \tG_N = \tG_N' \right]\right] \textrm{is a tautology}\\[2pt]
      \Leftrightarrow
      \neg \forall \widetilde{\tv}. \left[\TC \Rightarrow \exists \widetilde{\tv'}. \left[\TC' \land 
      \tG_1 = \tG_1' \land \cdots \land \tG_N = \tG_N' \right]\right] \textrm{is unsat}
      \\[2pt]
      \Leftrightarrow \exists \widetilde{\tv}. \left[\TC \land \forall \widetilde{\tv'}. \left[ \TC' \Rightarrow \neg \left[ 
      \tG_1 = \tG_1' \land \cdots \land \tG_N = \tG_N' \right] \right]\right] \textrm{is unsat}\\
    \end{array}
    \]
\end{small}
In our implementation, we use the SMT solver Yices for solving this formula (which is decidable~\cite{duterte15smt}), where all time variables have type Reals. 

\begin{lemma}
Determining whether two timed observables are equivalent is decidable.
\end{lemma}

Together with Proposition~\ref{prop:finite-traces}, we obtain the decidability of timed observational equivalence checking.

\begin{theorem}
Determining whether two configurations are time observationally equivalent is decidable.
\end{theorem}

%% file: experiments.tex
\begin{figure}
\begin{center}
  \includegraphics[width=0.35\textwidth]{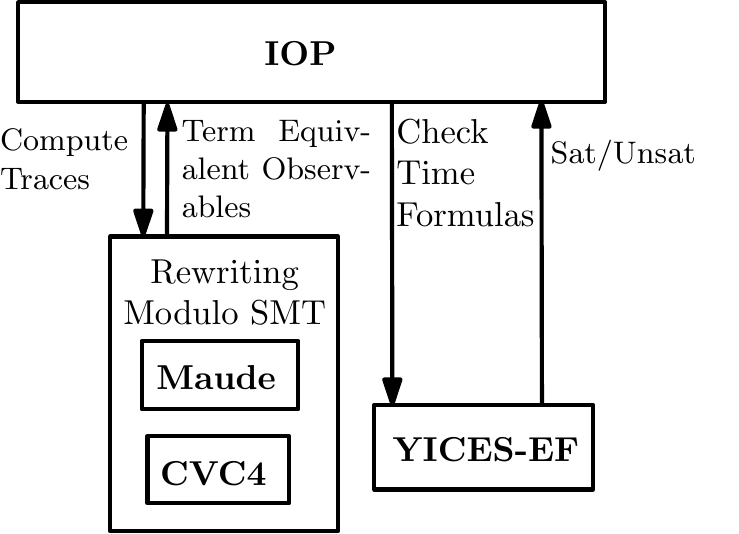}
\end{center}
  \caption{Timed Observational Equivalence Solver Architecture.}
  \label{fig:solver}
  \vspace{-2mm}
\end{figure}
We implemented a tool that checks for timed observational equivalence. Its architechture is depicted in Figure~\ref{fig:solver}. It is constructed over the tools:
\begin{itemize}
  \item \textbf{Maude:} We implemented in Maude all the machinery necessary specifying timed protocols as well as checking the term equivalence of observables. Moreover, since its alpha 111 version, Maude allows to call the CVC4 SMT solver~\cite{barrett11cvc4} from inside Maude programs. This allows the implementation of Rewriting Modulo SMT by using conditional rewrite rules that are only allowed to rewrite if the resulting constraint set is satisfiable. For our applications, this means the set of time constraints;

  \item \textbf{YICES-EF:} Since CVC4 does not provide the API for checking $\exists \forall$ formulas needed for proving time equivalence, we integrated our Maude machinery with the solver YICES-EF which can check for the satisfiability of such formulas. This integration has been carried out by using the IOP framework~\cite{mason-talcott-04wrla}.
\end{itemize}

Verification is coordinated by the IOP implementation. Given two initial configuration $\Cscr$ and $\Cscr'$ to be checked for their timed observational equivalence, IOP sends a command to the Maude+CVC4 tool to enumerate all obsevables for $\Cscr$ and $\Cscr'$, computes for each observable $\Oscr$ of $\Cscr$ the set of term-equivalent observables $\{\Oscr_1',\ldots, \Oscr_n'\}$ of $\Cscr'$ and vice-versa. If Maude finds some observable of $\Cscr$ that does not have at least one matching pair, that is, $n = 0$, $\Cscr$ and $\Cscr'$ are not equivalent. (Similarly for some pair of $\Cscr'$.) Otherwise, we continue by the timing equivalence condition. IOP attempts to find at least one time equivalent observable $\Oscr_i'$ for $\Oscr$. If it does not find it, then $\Cscr$ and $\Cscr'$ are not equivalent. It does this checking by building the formula as described in Section~\ref{subsec-timing} for checking for the timing equivalence of $\Oscr$ with $\Oscr_i'$. If YICES-EF returns Unsat, they are equivalent. Otherwise they are not equivalent and IOP tries with another observable.

\paragraph{Experimental Results} We carried out the following experiments:





\begin{itemize}
  \item \textbf{Red Pill Example}: Consider the timed protocol specified in Example~\ref{ex:red-pill}. We checked whether it is possible for an intruder to distinguish whether an application is running over a virtual machine or not. That is, we checked whether an initial configuration with a player running an application over a virtual machine is timed equivalent to the initial configuration with a player running the same application over a non-virtual machine. 

  \item \textbf{Passport Example}: Consider the timed protocol specified in Example~\ref{ex:passport}. We checked whether the intruder can distinguish the following two configurations both with two protocol sessions: the first where both protocol sessions are carried out with the same passport and the second where the protocol sessions are carried out with different passports. 

  \item \textbf{Corrected Passport Example}: We additionally considered a modification of the Passport example where the timed protocol is corrected in the sense that it sends necessarily both error messages at the same time. 

  \item \textbf{Anonymous Protocol}: Consider the timed protocol specified in Example~\ref{ex:anonymous}. We checked whether it is possible for an intruder to distinguish whether two players belong to the same group or not.  That is, we checked whether the initial configuration with a player that receives a message from a member of the same group is timed equivalent to the initial configuration with a player that receives a message from a player of a different group.
\end{itemize}

\begin{table}[t]
\begin{center}
\begin{tabular}{@{\quad}c@{\quad}|@{\quad}c@{\quad}|@{\quad}c@{\quad}|@{\quad}c@{\quad}}
\toprule
\textbf{Scenario} & \textbf{Result} & \textbf{Observables} & \textbf{States} \\
\midrule
Red-Pill & Not Equiv & 19/19 & 74/74\\
Passport & Not Equiv & 36/27 & 138/112\\
Passport-Corrected & Equiv & 36/27 & 138/112\\
Anonymous & Not Equiv & 2/3 &  7/9\\
\bottomrule
\end{tabular}
\end{center}
\caption{Experimental Results. Each experiment involves the proving the timed observational equivalence of two configurations. It contains number of observables (traces) for each configuration and the total number of states in the whole search tree required to traverse to enumerate all observables.}
\label{tab:exp}
\vspace{-8mm}
\end{table}

Table~\ref{tab:exp} summarizes the results of our experiments. Our tool was able to (correctly) identify the cases when the given configurations are timed observational equivalent. More impressive, however, is the number of states and observables it needed to traverse for doing so. In all experiments the number of states in the whole search tree was less than 140 states and the number of observables were less than 40. This is a very small number when compared to usual applications in Maude (which can handle thousands of states even when using Rewriting Modulo SMT~\cite{nigam16esorics}). This demonstrate the advantage of representing timing symbolically. As expected the number of observable for the passport example were greater as its configurations had two protocol session, while in the remaining experiments configurations have only one protocol session. Finally, since the number of observables was small, the number of calls to Yices was small and therefore, verification for all experiments took less than some seconds.

%% file: related.tex
This paper introduced a novel definition of timed equivalence for security protocols using symbolic time constraints. We demonstrated how symbolic time equivalence can be proved automatically with the use of Rewriting Modulo SMT and existing SMT-solvers. The combination of such constraints with Rewriting Modulo SMT greatly reduces the number of states required to enumerate all traces. We implemente the machinery for proving the timed observational equivalence and and showed experimentally with some proof-of-concept examples that our technique is practical.

For future work, we will be integrating the machinery developed here with the Maude libraries~\cite{duran16ijcar} which can be used for the verification of security protocols that use weaker notions of encryption. We are also interested in integrating the machinery developed with Narrowing and Maude-NPA.

\paragraph{Related Work:} The literature on symbolic verification is vast~\cite{basin-etal-2004jis,cortier09csf,cervesato99csfw,escobar07fosad,cortier09csf,gazeau17esorics}. However, most of this work uses symbolic reasoning for proving reachability properties. 

One exception is the work of~\cite{cortier09csf}. Indeed, for the observational equivalence involving terms, we have been heavily inspired by~\cite{cortier09csf}, but there are some differences. The main difference is that our timed protocols includes both time symbols and branching. Also, we also implemented our machinery for term equivalence in Maude and use SMT-solvers for search (Rewriting Modulo SMT) and proving the timing equivalence.

Cheval and Cortier~\cite{CC-post15} propose a definition of timed equivalence reducing it to other notions of equivalence taking into account the length of messages. We take a different approach by using timed constraints and SMT-Solvers. This allows us to relate time symbols using inequalities, \eg, $\tv_1 \geq \tv_2 + 10$. 
While a more detailed comparison is left to future work, the use time symbols allows the use of off-the-shelf SMT-solvers which are constantly improving. Indeed, as reported in~\cite{CC-post15}, the corrected version of the Passport Example did not terminate after two days of experiments. 

The recent work~\cite{gazeau17esorics} demonstrates how to automate the prove of observational equivalence of protocols that may contain branching and xor. While we allow for branching, we do not consider theories involving xor. However, we do consider timing aspects, which is not considered in~\cite{gazeau17esorics}. Thus these work are complementary. As described above, we expect in the future to support xor (and other equational theories) by using the built-in Maude matching and unification functionality~\cite{duran16ijcar}.

Finally, there have been other frameworks for the verification of timing properties of systems~\cite{bella98esorics,evans00esorics,gorrieri03esop,jakubowska07fi,kanovich12rta}. A main difference is that the properties verified were reachability properties and not timed equivalence. 